\documentclass[aps,prl,superscriptaddress,longbibliography,showpacs,twocolumn]{revtex4-2}

\usepackage{pdfpages}

\makeatletter
\AtBeginDocument{\let\LS@rot\@undefined}
\makeatother

\usepackage{soul}
\setstcolor{red}
\usepackage{comment}

\usepackage{lipsum}
\usepackage{caption}
\usepackage{tikz}
\usetikzlibrary{arrows.meta}

\DeclareCaptionJustification{justified}{\leftskip=0pt \rightskip=0pt \parfillskip=0pt plus 1fil}

\usepackage{graphicx}% Include figure files
\usepackage{dcolumn}% Align table columns on decimal point
\usepackage{bm}% bold math
\usepackage{caption}
\usepackage{amsmath,amssymb,amsthm}
\usepackage{braket}
\usepackage{empheq}
\usepackage{subfig}
\usepackage{tikz}
\usepackage{url}
\usepackage{hyperref}
\usepackage{quantikz}
\usepackage{xcolor}
\usepackage{float}
\DeclareMathOperator{\Tr}{Tr}
\DeclareMathOperator{\poly}{poly}

\usepackage{algorithm}
\usepackage[noend]{algpseudocode}% http://ctan.org/pkg/algorithmicx

\hypersetup{
    breaklinks=true,
    colorlinks=true,       % false: boxed links; true: colored links
    linkcolor=purple,          % color of internal links
    citecolor=purple,        % color of links to bibliography
    filecolor=magenta,      % color of file links
    urlcolor=purple,           % color of external links
    runcolor=purple
}

\newtheorem{theorem}{Theorem}
\newtheorem*{theorem*}{Theorem}
\newtheorem{corollary}{Corollary}%[theorem]
\newtheorem{lemma}[theorem]{Lemma}

\usepackage{etoolbox}
\makeatletter
\patchcmd{\@citex}{\unskip, et al.}{\unskip et al.}{}{}
\makeatother
\begin{document}
\title{High-order dynamical decoupling in the weak-coupling regime}
\author{Leeseok Kim}
\affiliation{Center for Quantum Information and Control and Department of Electrical and Computer Engineering, University of New Mexico, NM 87131, USA}

\author{Milad Marvian}
\affiliation{Center for Quantum Information and Control and Department of Electrical and Computer Engineering, University of New Mexico, NM 87131, USA}

\begin{abstract}
We introduce a high-order dynamical decoupling (DD) scheme for arbitrary {bounded} system-bath interactions in the weak-coupling regime. Given any decoupling group $\mathcal G$ that averages the interaction to zero, our construction {guarantees the existence of pulse sequences with at most $(|\mathcal G|-1)K$ pulses}, while canceling all error terms linear in the system-bath coupling strength up to order $K$ in the total evolution time. As a corollary, for an $n$-qubit system with $k$-local system-bath interactions, we obtain an $\mathcal{O}(n^{k-1}K)$-pulse sequence, a significant improvement over existing schemes with $\mathcal{O}(\exp(n))$ pulses (for $k=\mathcal{O}(1)$). The construction is obtained via a mapping to the continuous necklace-splitting problem, which asks how to cut a multi-colored interval into pieces that give each party the same share of every color. We provide explicit pulse sequences for suppressing general single-qubit decoherence, prove that the pulse count is asymptotically optimal, and verify the predicted error scaling in numerical simulations. For the same number of pulses, we observe that our sequences outperform the state-of-the-art Quadratic DD in the weak-coupling regime. {We also construct explicit high-order sequences for suppressing representative 2-local noise models.} {Finally, the same construction extends to suppress slow, time-dependent classical noise and to filter-function design.}
\end{abstract}

\maketitle
\textit{Introduction---} Dynamical decoupling (DD) \cite{viola1999dynamical} is a powerful technique for mitigating decoherence caused by unavoidable system-bath interactions. The central idea is to apply a sequence of fast control pulses to the system, effectively averaging out the coupling between the system and its environment. Due to its effectiveness and low overhead, DD has been extensively explored both theoretically \cite{hahn1950spin,meiboom1958modified,maudsley1986modified,viola2003robust,viola2005random,khodjasteh2005fault,khodjasteh2007performance,uhrig2007keeping,yang2008universality,uhrig2009concatenated,west2010high,west2010near,uhrig2010rigorous,xia2011rigorous,ng2011combining,wang2011protection,jiang2011universal,quiroz2013optimized,bookatz2016improved,genov2017arbitrarily,brown2025efficient,yi2026faster} and experimentally \cite{witzel2007concatenated,biercuk2009experimental,biercuk2009optimized,alvarez2010performance,lange2010universal,ryan2010robust,barthel2010interlaced,ajoy2011optimal,souza2011robust,bylander2011noise,medford2012scaling,zhao2012decoherence,xu2012coherence,farfrunik2015optimizing,pokharel2018demonstration,tripathi2022suppression,ezzell2023dynamical,evert2025syncopated,kasatkin2026quantum} across a wide range of platforms, making it a key ingredient in the implementation of state-of-the-art quantum processors \cite{kim2023scalable,acharya2024quantum,paetznick2024demonstration,bluvstein2025a,vezvaee2025surface,bluvstein2025a}.

Since the 1950s, numerous DD sequences have been proposed to suppress decoherence more effectively. These designs typically aim to increase the cancellation order in the total evolution time $T$, but at the cost of {a larger number of pulses}. For the most general case of single-qubit system-bath coupling, two high-order DD protocols have been developed. The earliest is Concatenated DD (CDD) \cite{khodjasteh2005fault}, whose number of pulses grows exponentially with the cancellation order $K$. This was followed by Quadratic DD (QDD) \cite{west2010near} which achieves the same decoupling order while requiring only a quadratic number of pulses in $K$.

\begin{figure}[t!]
    \centering
    \includegraphics[width=1\linewidth]{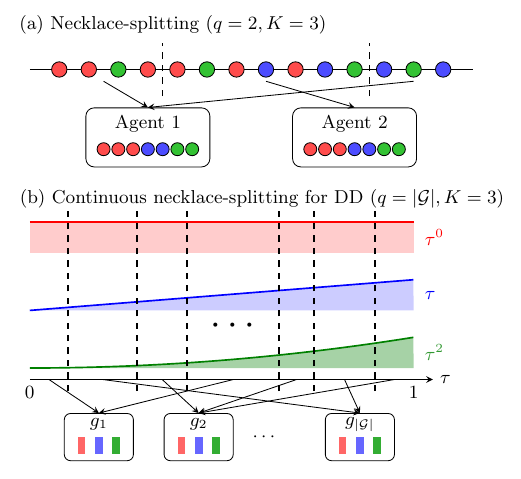}
   \caption{(a) Discrete necklace-splitting example with two agents ($q=2$) and three colors ($K=3$): by choosing two cuts, each agent receives the same number of red, blue, and green beads. (b) Continuous necklace-splitting picture for dynamical decoupling: the first three time-moments $\int_{0}^{1}\tau^{m}d\tau$ ($K=3$ in this example) are partitioned by the cuts and assigned to the $|\mathcal G|$ agents ($q=|\mathcal G|$) so that each agent receives the same moments. In the DD construction, the cuts mark the pulse times, $\mathcal G$ denotes the decoupling group for the given noise model, and the time-moments correspond to the error terms, whose equal distribution ensures that averaging over the group cancels these contributions. The necklace-splitting theorem guarantees that there exist at most $(q-1)K$ such cuts.}
\label{necklace-figure}
\end{figure}

For $n$-qubit systems, both CDD and the multi-qubit generalization of QDD, Nested Uhrig DD (NUDD) \cite{wang2011protection,jiang2011universal}, can suppress general decoherence. However in many physical settings the dominant noise is $k$-local (often with a constant $k$), meaning it affects only a small subset of qubits at a time~\cite{terhal2005fault,aharonov2006fault,ng2011combining,von2020two}. {In this case, first-order DD sequences of length $\mathcal O(n^{k-1})$ were constructed from orthogonal arrays, reflecting noise locality~\cite{bookatz2016improved}.} To date, however, there is no high-order DD construction that exploits the $k$-local structure of the noise. {Although CDD and NUDD apply to $k$-local models, they ignore $k$ and thus substantially overestimate the required pulse count.}

\begin{table*}[t!]
\centering
\renewcommand{\arraystretch}{1.2}
\setlength{\tabcolsep}{9pt}
\begin{tabular}{lccc}
\hline
Scheme & System / noise model & Leading error scaling & Number of pulses \\
\hline
\multicolumn{4}{c}{\textbf{Single-qubit, general noise}} \\[4pt]
CDD \cite{khodjasteh2005fault}
  & 1-qubit, general 
  & $\mathcal O(JT^{K+1}) + \mathcal O(J^2 T^2)$ 
  & $\mathcal O(4^K)$ \\
QDD \cite{west2010near}
  & 1-qubit, general 
  & $\mathcal O\big(JT^{K+1}\big) + \mathcal O\big(J^2T^{K+1}\big)$ 
  & $\mathcal O(K^2)$ \\[2pt]
This work 
  & 1-qubit, general 
  & $\mathcal O(JT^{K+1}) + \mathcal O(J^2 T^2)$ 
  & $\mathcal O(K)$ \\[4pt]
\multicolumn{4}{c}{\textbf{$n$-qubit, general noise}} \\[4pt]
CDD \cite{khodjasteh2005fault}     & $n$ qubits, general 
         & $O(JT^{K+1}) + O(J^2 T^2)$ 
         & $\mathcal O (4^{nK})$ \\[2pt]
NUDD \cite{wang2011protection,jiang2011universal}    & $n$ qubits, general 
         & $\mathcal O(JT^{K+1}) + \mathcal O(J^2 T^{K+1})$ 
         & $\mathcal O(K^{2n})$ \\[2pt]
This work
         & $n$ qubits, general 
         & $\mathcal O(JT^{K+1}) + \mathcal O(J^2 T^2)$ 
         & $\mathcal O(4^{n} K)$ \\[4pt]
\multicolumn{4}{c}{\textbf{$n$-qubit, $k$-local noise}} \\[4pt]
OA-based DD \cite{bookatz2016improved}
         & $n$ qubits, $k$-local 
         & $\mathcal O(JT^{2}) + \mathcal O(J^2 T^{2})$ 
         & $\mathcal O(n^{k-1})$ \\[2pt]
{This work} 
         & $n$ qubits, $k$-local 
         & $\mathcal O(JT^{K+1}) + \mathcal O(J^2 T^2)$ 
         & $\mathcal O(n^{k-1} K)$ \\[4pt]
\multicolumn{4}{c}{{\textbf{$n$ qubits on a graph with chromatic number $\chi$, $2$-local noise}}} \\[4pt]
CHaDD \cite{brown2025efficient}
  & $n$ qubits, $2$-local
  & $\mathcal O(JT^{2}) + \mathcal O(J^2 T^2)$
  & $\mathcal O(\chi)$ \\[2pt]
This work
  & $n$ qubits, $2$-local
  & $\mathcal O(JT^{K+1}) + \mathcal O(J^2 T^2)$
  & $\mathcal O(\chi K)$ \\
\hline
\end{tabular}
\caption{Comparison of representative dynamical decoupling (DD) schemes, showing the noise model, leading error scaling, and pulse-count scaling as functions of cancellation order $K$, system size $n$, noise locality $k$, and, where applicable, the chromatic number $\chi$ of the underlying coupling graph. Abbreviations: CDD = Concatenated DD, QDD = Quadratic DD, NUDD = Nested Uhrig DD, OA = orthogonal-array, {CHaDD = chromatic Hadamard DD}. Our construction cancels all contributions linear in the system-bath coupling strength $J$ up to order $K$ in the total evolution time $T$ using a number of pulses linear in $K$, and applies to any noise model admitting a decoupling group $\mathcal G$. In particular, there exist choices of $\mathcal G$ with size $\mathcal O(1)$ for general single-qubit noise, $\mathcal O(4^n)$ for general $n$-qubit noise, $\mathcal O(n^{k-1})$ for $k$-local noise on $n$ qubits, and $\mathcal O(\chi)$ for $2$-local noise on a graph with chromatic number $\chi$.}
\label{dd-scaling}
\end{table*}

In this Letter, we present a high-order DD construction for arbitrary system-bath interactions in the weak-coupling regime, where the system-bath  interaction strength is  small compared to the intrinsic system and/or the bath energy scales (defined precisely later), a condition that is well motivated in many physical platforms~\cite{breuer2002theory}. This regime is standard in the DD literature, with typical benchmarks assuming  coupling  $10^{-2}$--$10^{-6}$ times the pure-bath energy scale~\cite{khodjasteh2005fault,west2010high,west2010near,uhrig2010rigorous,xia2011rigorous,ng2011combining,quiroz2013optimized,yi2026faster}. DD is also widely used to preserve the system’s intended dynamics while suppressing weaker environmental couplings---for instance, to maintain a desired system evolution to perform robust quantum computation \cite{viola1999universal,lidar2008towards,khodjasteh2008rigorous,west2010high,ng2011combining,quiroz2012high,xu2012coherence,de2013universal} or to preserve the sensing signal \cite{degen2017quantum}.

%Our scheme yields pulse sequences whose length scales linearly with both the cancellation order $K$ and decoupling group size $|\mathcal{G}|$, while cancelling all {error terms that are linear} in the system-bath coupling strength $J$ up to order $K$ in the total evolution time $T$. In the weak-coupling regime, where these first-order in $J$ contributions dominate the error, our sequences (1) achieve the same {leading-order scaling $\mathcal{O}(J T^{K+1})$} as CDD, QDD, and NUDD for single-qubit and general $n$-qubit decoherence, but with substantially fewer pulses, and (2) provide, to our knowledge, the first high-order DD construction for $k$-local noise whose pulse complexity explicitly reflects the locality, scaling as $\mathcal{O}(n^{k-1}K)$. These results are summarized in Table \ref{dd-scaling}. In particular, for single-qubit general decoherence we match the leading error scaling of QDD using only $\mathcal{O}(K)$ pulses, and we further prove that this linear scaling is asymptotically optimal. Outside this weak-coupling regime, however, the residual second-order contribution in our scheme scales only as $J^2 T^2$. Hence, protocols such as QDD and NUDD, which suppress all higher-order terms in $T$ at every order in $J$, can outperform our construction in the strong-noise regime.

{Our construction uses $\mathcal{O}(|\mathcal G|K)$ pulses to cancel all $\mathcal O(J)$ terms through order $K$ in $T$, yielding error $\mathcal O(JT^{K+1})+\mathcal O(J^2T^2)$. In the weak-coupling regime, where the linear-in-$J$ contributions dominate the error, our construction achieves the same leading error scaling as QDD for general single-qubit noise and as CDD and NUDD for general $n$-qubit noise, with fewer pulses. For $k$-local noise, a decoupling group of size $\mathcal O(n^{k-1})$ gives $\mathcal O(n^{k-1}K)$ pulses, providing, to our knowledge, the first high-order DD construction whose pulse count reflects noise locality. These results are summarized in Table~\ref{dd-scaling}. For general single-qubit noise, the $\mathcal O(K)$ scaling is asymptotically optimal. Outside weak coupling, however, QDD and NUDD can perform better because they suppress higher orders in $T$ at every order in $J$.}

We first show that a pulse sequence of length at most $(|\mathcal G|-1)K$ {must exist} by relating the pulse-scheduling task to the \emph{necklace-splitting problem} \cite{alon1987splitting}, introduced in the 1980s and widely studied in fair division and combinatorics. In its simplest form, the necklace-splitting problem asks: given a necklace with beads of several colors, can one cut the necklace into pieces and distribute them among a fixed number of players so that each player receives exactly the same quantity of every color? This fair-division framework provides a natural abstraction for our construction: the different bead colors correspond to different cancellation conditions that the DD sequence must satisfy, while assigning pieces to players corresponds to choosing the appropriate toggling (or pulse) operators (see Fig.~\ref{necklace-figure}). Moreover, by leveraging known efficient approximation algorithms for the necklace-splitting problem \cite{alon2021efficient}, such sequences can be constructed explicitly in {$\mathcal O(\poly(|\mathcal{G}|,K))$ time}. %Finally, we numerically obtain explicit DD pulse sequences of length $3K$ for the 1-local noise model, as well as explicit higher-order DD sequences for 2-local noise on several graphs, and confirm their predicted error scaling through simulations. In the weak-coupling regime, our 1-local sequences can outperform QDD using the same number of pulses, while the $2$-local constructions extend CHaDD to the higher-order setting.
{We numerically construct explicit $3K$-pulse high-order DD sequences for 1- and 2-local noise on several graphs, confirming their predicted error scaling. In weak coupling, the former can outperform QDD at equal pulse count, while the latter extend CHaDD~\cite{brown2025efficient} to higher order.}

{\textit{Setting, control, and moment conditions---}} We consider an $n$-qubit system $S$ coupled to a bath $B$ with
\begin{align}
H(t)  = {H_S \otimes I_B} + I_S \otimes H_B  + H_{SB} + H_{\mathrm{c}}(t) \otimes I_B,
\label{model}
\end{align}
{where $H_S$, $H_B$, $H_{SB}$, and $H_{\mathrm{c}}(t)$ denote the internal system, bath, system-bath interaction, and control Hamiltonians, respectively.} We assume $\|H_0\| = 1$ and $\|H_{SB}\| = J$, where $H_0 := H_S \otimes I_B + I_S \otimes H_B$. More generally, if $\|H_0\| = \beta$, one can recover this normalization by measuring time in units of $\beta^{-1}$.%, equivalently $T\mapsto \beta T$ and $J\mapsto J/\beta$.  

Control is modeled as a sequence of $L$ instantaneous Pauli pulses {$P_1,\ldots,P_L$} applied at interior times $0<t_1<\cdots<t_L<T$, with $t_0:=0$ and $t_{L+1}:=T$. The corresponding control propagator is piecewise constant, $U_c(t)=P_j\cdots P_1$ for $t\in[t_j,t_{j+1})$, with $U_c(0)=I_S$. Additional pulse may be applied at $t=T$ to enforce cyclicity ($U_c(T)=I_S$); since it acts only at the endpoint, it does not affect the switching functions on $[0,T)$.

{Assuming the standard DD condition $[H_0,H_{\mathrm c}(t)]=0$ for all $t$ (see, e.g., Refs.~\cite{viola1999dynamical,khodjasteh2005fault,yi2026faster}), let $U_0(t)$ be the propagator generated by $H_0+H_{\mathrm c}(t)\otimes I_B$,}
\begin{align}\label{U0}
{U_0(t)=U_c(t)e^{-iH_0 t}.}
\end{align}
{Moving to the interaction picture with respect to $U_0(t)$, one may write the full propagator as}
\begin{align}\label{U_interaction}
U(t)=U_0(t)\mathcal{T}\exp\left(-i\int_0^t \widetilde H(t') dt' \right),
\end{align}
where $\widetilde H(t):=U_0^\dagger(t)H_{SB}U_0(t)$.

Using Eq.~\eqref{U0} and expanding $H_{SB}=\sum_\alpha \sigma_\alpha\otimes B_\alpha$ where $\sigma_\alpha$ are (non-identity) $n$-qubit Pauli strings and $B_\alpha$ are Hermitian bath operators, the toggling-frame Hamiltonian is
\begin{align}
\widetilde H(t) = \sum_\alpha e^{iH_0t}\left(U_c^\dagger(t) \sigma_\alpha U_c(t) \otimes B_\alpha\right) e^{-iH_0t}.
\end{align}
Since $U_c(t)$ is a piecewise-constant Pauli operator, for each $\alpha$ we define the switching function $y_\alpha(t)\in\{\pm1\}$ by $U_c^\dagger(t) \sigma_\alpha U_c(t) = y_\alpha(t) \sigma_\alpha$. Hence, 
\begin{align}
\widetilde H(t) = \sum_\alpha y_\alpha(t) e^{iH_0t}\left(\sigma_\alpha \otimes B_\alpha\right)e^{-iH_0t}.
\label{toggling-U0}
\end{align}

{Introduce normalized time $\tau=t/T$ and denote $y_\alpha(\tau):=y_\alpha(T\tau)$. We define the generalized moments
\begin{align}
M_{\alpha,m}:=\int_0^1y_\alpha(\tau)\tau^m d\tau.
\label{MomentsDef}
\end{align}
Expanding the $H_0$ evolution in Eq.~\eqref{toggling-U0} in powers of $t$ shows that the coefficient at order $T^{m+1}$ in the first Magnus term is proportional to $M_{\alpha,m}$. Consequently, imposing
\begin{align}
M_{\alpha,m}=0,
\qquad
m=0,\ldots,K-1,
\label{moment-constraints-compact}
\end{align}
for every Pauli term $\sigma_\alpha$ appearing in $H_{SB}$ eliminates all linear-in-$J$ contributions through order $T^K$. Assuming convergence of the Magnus series, the remaining contributions start at second order in $J$, yielding
\begin{align}
\|U(T)-U_0(T)\| = \mathcal O(JT^{K+1})+\mathcal O(J^2T^2),
\label{target}
\end{align}
with the derivation provided in the End Matter. We call the parameter range in which the linear-in-$J$ contribution dominates the \emph{weak-coupling regime}; in the present units, this corresponds to $J=\mathcal O(T^{K-1})$.}

\textit{Main result---} Let $\mathcal{G} \subset \mathcal{P}_n$ be a finite subgroup of the $n$-qubit Pauli group. We say $\mathcal{G}$ is a \emph{decoupling group} \cite{viola1999dynamical,zanardi1999symmetrizing,yi2026faster} if every element commutes with $H_0$ and the group twirl of $H_{SB}$ vanishes:
\begin{align}
\forall g \in \mathcal G, ~~[H_0, g \otimes I_B] = 0, \quad \Pi_{\mathcal G}(H_{SB}) = 0,
\label{decoupling-group-condition}
\end{align} where $\Pi_{\mathcal G}(X) := \frac{1}{|\mathcal G|}\sum_{g \in \mathcal G} (g^\dagger \otimes I_B) X (g \otimes I_B)$. 

Our central result shows that the order-$K$ moment constraints in Eq.~\eqref{moment-constraints-compact} can always be met using at most $(|\mathcal G|-1)K$ pulses.

\begin{theorem}[Existence of an order-$K$ moment-cancelling sequence with at most $(|\mathcal G|-1)K$ pulses]\label{thm:existence}
{Let $\mathcal G$ be a decoupling group as defined in Eq.~\eqref{decoupling-group-condition}. For any integer $K\ge 1$, there exist piecewise-constant functions $y_\alpha:[0,1]\to\{\pm1\}$, indexed by all Pauli strings $\sigma_\alpha$ appearing in $H_{SB}$, such that:}
\begin{enumerate}
\item\label{cond:i}
there exist (normalized) pulse timings $0 < \tau_1 < \cdots < \tau_L < 1$ with
\begin{align}\label{condition1}
L \le (|\mathcal{G}|-1)K 
\end{align}
for which all $y_\alpha(\tau)$ are constant on each subinterval $[0,\tau_1), [\tau_1,\tau_2), \ldots, [\tau_L,1]$, and 
\item\label{cond:ii} {the order-$K$ moment constraints in Eq.~\eqref{moment-constraints-compact} hold.}
\end{enumerate}
\end{theorem}
{The full proof of Theorem~\ref{thm:existence} is provided in the Supplemental Material (SM)~\cite{suppmat}. Here we give the intuition and explain the role of necklace splitting. View normalized time $\tau\in[0,1]$ as a continuous necklace whose point $\tau$ carries $K$ ``colors'' with weights $1,\tau,\ldots,\tau^{K-1}$ (cf.~Fig.\ref{necklace-figure}~(b)). The continuous necklace-splitting theorem~\cite{alon1987splitting} (Lemma~\ref{lem:necklace} in SM) guarantees that with at most $(|\mathcal G|-1)K$ cuts we can distribute the pieces among $|\mathcal G|$ bins labeled by $g\in\mathcal G$ so that every bin receives exactly the same amount of every color. Assigning to each bin the corresponding toggling-frame sign pattern (from conjugation by $g$) makes all low-order time weights identical across bins, and the decoupling-group condition implies these signed contributions sum to zero. Hence all moments up to order $K-1$ vanish.}

Given a labeled partition $0=t_0<t_1<\cdots<t_{L}<t_{L+1}=T$ with segment labels $g_\ell\in\mathcal G$ ($\ell=0,\ldots,L$) and $g_0=I$, the physical pulse sequence is $P_\ell := g_\ell g_{\ell-1}^\dagger$, so that $U_c(t)=g_\ell$ for $t\in[t_\ell,t_{\ell+1})$. 

{Theorem~\ref{thm:existence} only guarantees that such a pulse-efficient sequence exists. For any $\varepsilon>0$, however, a deterministic algorithm constructs an approximate sequence in $\mathcal O\left(\poly(|\mathcal G|,K,\log(1/\varepsilon))\right)$ time, using $\mathcal O\left(|\mathcal G|K\log(|\mathcal G|/\varepsilon)\right)$ pulses and satisfying $|M_{\alpha,m}|\le\varepsilon/(m+1)$ for every $\alpha$ and $m<K$; see the End Matter.}

\begin{figure}[t!]
  \centering
  \includegraphics[width=8.6cm]{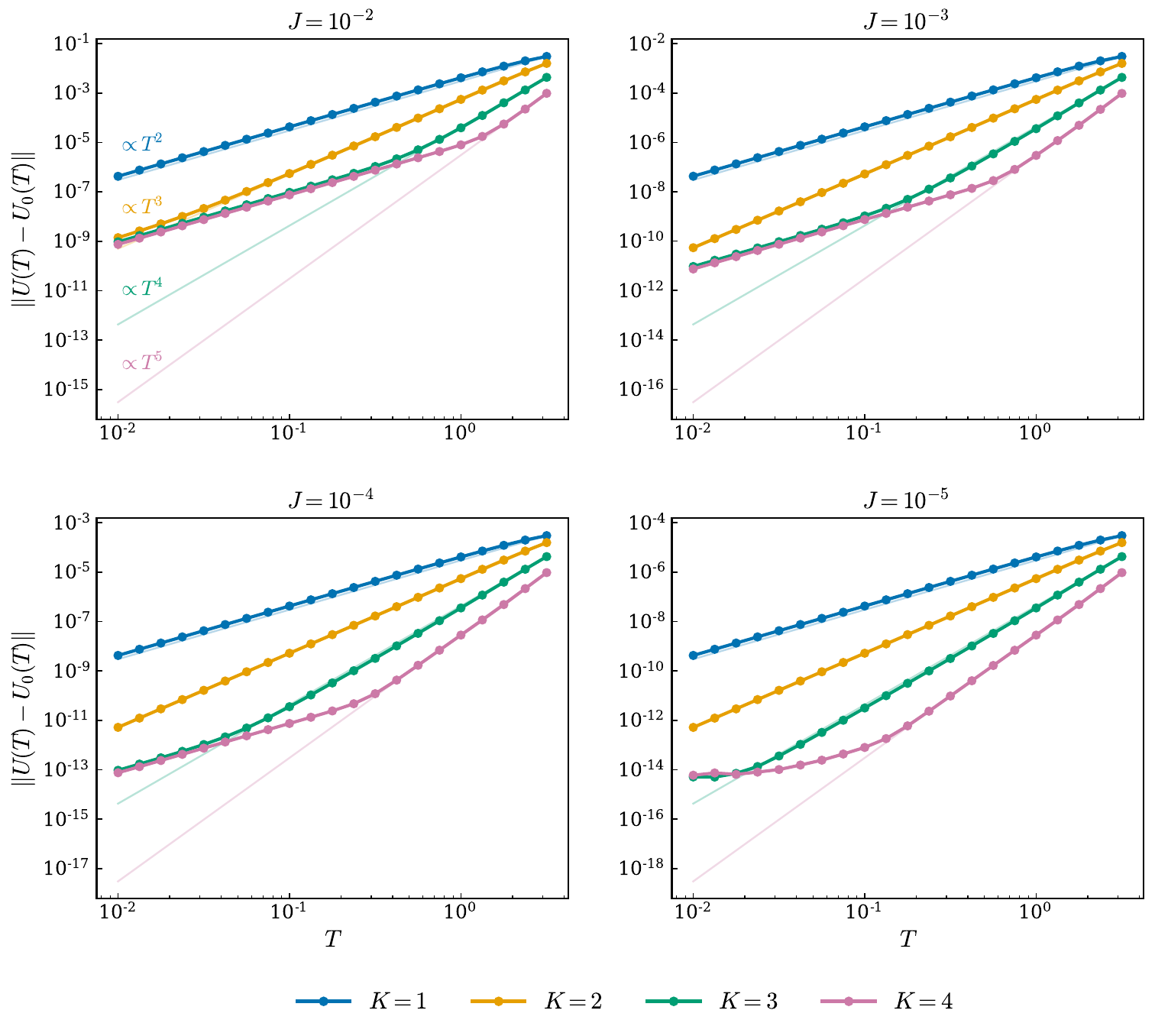}
  \caption{{Errors for optimized DD sequences at different coupling strengths. Here $\|H_0\|=1$, and $T$ denotes the total evolution time. As $J$ decreases, the error increasingly follows the predicted leading behavior $\mathcal{O}(JT^{K+1})$ for the order-$K$ sequence. The observed crossover reflects the competition between the two terms in \eqref{target}: for sufficiently small $T$, the $\mathcal{O}(J^2T^2)$ contribution can dominate, giving an apparent slope of $2$, whereas for larger $T$ the $\mathcal{O}(JT^{K+1})$ term becomes dominant. The crossover occurs around $T\sim J^{1/(K-1)}$, where the two contributions become comparable. In this larger-$T$ regime, the higher-order scaling is confirmed by the fact that the curves become parallel to the shaded guide lines, which indicate the power-law scalings $T^2,T^3,T^4,$ and $T^5$.}}
\label{fig:Jcomparison}
\end{figure}

\textit{Implication: $k$-local noise model---} Suppose that $H_{SB}$ is $k$-local, in the sense that every Pauli term $\sigma_{\alpha}$ appearing in $H_{SB}$ has weight at most $k$. In this setting, {one can construct a decoupling group} $\mathcal G$ of size $|\mathcal G| = \mathcal{O}(n^{k-1})$ \cite{bookatz2016improved}. Theorem~\ref{thm:existence} thus guarantees the existence of pulse sequences with length $L = \mathcal{O}(n^{k-1} K)$, for which $M_{\alpha,m} = 0$ for all $m = 0,1,\dots,K-1$ and all Pauli strings $\sigma_\alpha$ of weight at most $k$. To the best of our knowledge, this is the first construction of \text{high-order} DD sequences whose pulse count explicitly reflects the locality of the noise, in contrast to generic high-order schemes such as CDD \cite{khodjasteh2005fault} and NUDD \cite{jiang2011universal}, whose cost does not exploit $k$-local structure. For $k=\mathcal{O}(1)$, our DD uses $\mathcal O(\poly(n))$ pulses, whereas CDD and NUDD use $\mathcal O(\exp(n))$ pulses; thus exploiting locality yields an exponential improvement.

For $1$-local noise, the best-known sequence is QDD \cite{west2010near}, which achieves error suppression of order $\mathcal{O}(JT^{K+1}) + \mathcal{O}(J^2 T^{K+1})$ using $\mathcal{O}(K^2)$ pulses. Consequently, in the weak-coupling regime where the $\mathcal{O}(J)$ term dominates, our construction that uses only $\mathcal O (K)$ pulses provides a quadratic reduction in the required number of pulses. {Moreover, this linear scaling in $K$ is optimal: any DD sequence satisfying $M_{\alpha,m}=0$ for all $\alpha$ and $m=0,\ldots,K-1$ must use at least $K$ pulses.}

{\begin{lemma}[$\Omega(K)$ lower bound]\label{lem:lowerbound}
Let $y:[0,1]\to\{\pm1\}$ be piecewise constant with $r$ sign changes on $(0,1)$. If 
\begin{align}\label{lemma_condition}
\int_0^1 y(\tau) \tau^m d\tau=0\qquad (m=0,1,\ldots,K-1),
\end{align}
then necessarily $r\ge K$.
\end{lemma}
The proof of Lemma~\ref{lem:lowerbound} is provided in the SM. Since single-axis noise is a special case of general noise, the same lower bound applies here, and thus the $\mathcal{O}(K)$ pulse scaling is optimal.}

\textit{Extension to other noise models---}
%Environmental fluctuations can also be modeled as time-dependent Hamiltonian noise and suppressed by DD \cite{green2013arbitrary,bylander2011noise,alvarez2010performance,souza2011robust,ramon2022qubit}. For a single qubit with $H(t)=\sum_{\alpha\in\{x,y,z\}}\beta_\alpha(t)\sigma_\alpha+H_c(t)$, where $\beta_\alpha(t)$ are unknown {$K$-times differentiable} noise amplitudes and $H_c(t)$ generates Pauli pulses, the leading toggling-frame Magnus term is determined by the generalized moments $M_{\alpha,m}$. Hence, imposing the moment constraints in \eqref{moment-constraints-compact} cancels the leading error through order $K$. In the slow-noise regime with cutoff $\omega_c$, this yields $\mathcal O\big(\beta_{\max}T(\omega_cT)^K\big)+\mathcal O(\beta_{\max}^2T^2)$ error, where $\beta_{\max}:=\sup_{t\in[0,T]}\sum_\alpha |\beta_\alpha(t)|$. 
{The same moment conditions apply to slowly varying classical Hamiltonian noise $H_n(t)=\sum_{\alpha\in\{x,y,z\}}\beta_\alpha(t)\sigma_\alpha$ ~\cite{green2013arbitrary,bylander2011noise}. For $K$-times differentiable noise with effective cutoff frequency $\omega_c$, they yield error $\mathcal O\left(\beta_{\max}T(\omega_cT)^K\right) +\mathcal O(\beta_{\max}^2T^2)$, where $\beta_{\max}:=\sup_t\sum_\alpha|\beta_\alpha(t)|$.}

{Our construction requires bounded $H_{SB}$ and therefore does not directly apply} to unbounded baths such as bosonic baths. Such settings are often treated using filter-function methods \cite{cywinski2008how,uhrig2008exact,paz2014general,clausen2012task,gordon2007universal,kofman2004unified}. Within this framework, we show that the same constraints enforce a $K$th-order zero of the control Fourier amplitudes at low frequency, and hence a $2K$th-order suppression of the associated quadratic filter functions. This provides a frequency-domain connection to unbounded environments, including the spin-boson model discussed in the SM.

\begin{figure}[t!]
  \centering
  \includegraphics[width=8.6cm]{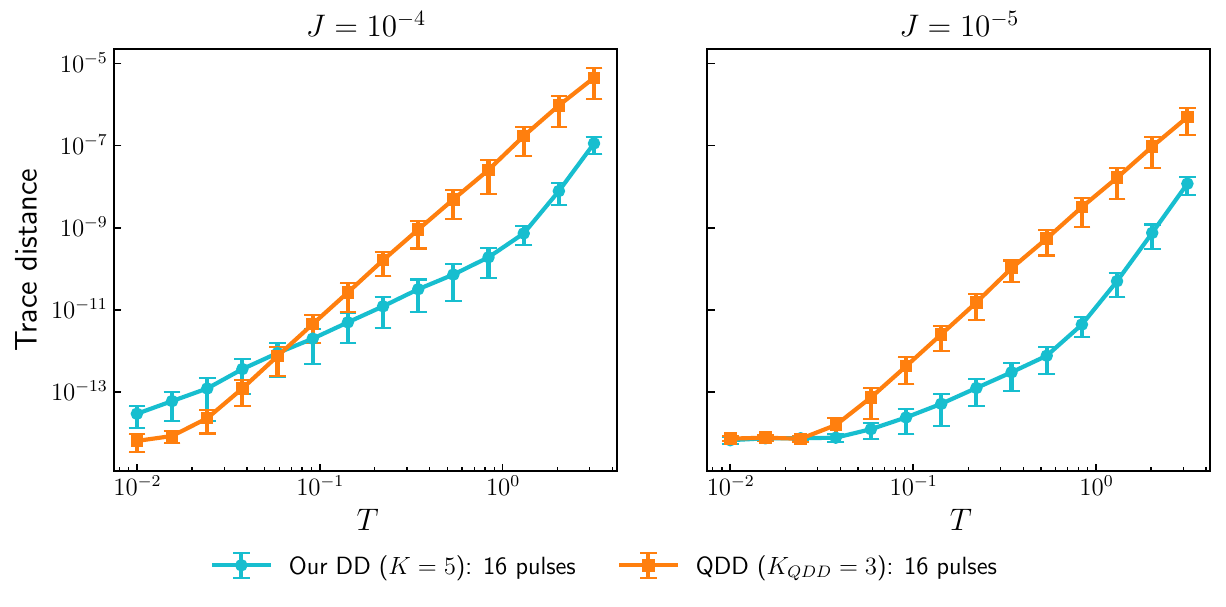}
  \caption{Comparison between our DD sequences and QDD for the single-qubit general decoherence model. We plot the average trace distance between the actual and ideal reduced system states versus $T$ for several weak coupling strengths $J$, comparing our $K=5$ sequence with $K_{\rm QDD}=3$ QDD, each using $16$ pulses. Each point is averaged over $100$ random product initial states.}
  \label{fig:QDD_comparison_main}
\end{figure}

\textit{Numerical simulation---} {We numerically test our construction for a single qubit coupled to a single-qubit bath, with $H_{SB}=\sum_{\alpha\in\{x,y,z\}}\sigma_\alpha\otimes B_\alpha$, $B_\alpha=\sum_{\mu\in\{x,y,z\}}c_{B_\alpha,\mu}\sigma_\mu$, and $H_B=\sum_{\alpha\in\{x,y,z\}}c_\alpha\sigma_\alpha$. All coefficients are chosen in $[0,1]$, and we normalize $\|H_B\|=1$. For $\mathcal G=\{I,X,Y,Z\}$, Theorem~\ref{thm:existence} guarantees a sequence with at most $3K$ interior pulses (one additional pulse at $t=T$ may be used to enforce cyclicity). Because the deterministic algorithm achieves $|M_{\alpha,m}|\le\varepsilon/(m+1)$ using $\mathcal O(K\log(1/\varepsilon))$ pulses, its logarithmic overhead can make the pulse count substantially larger than the existence bound $3K$ for the small values of $\varepsilon$ relevant to DD. We therefore fix the alternating pattern $X\to Z\to X\to Z\to\cdots$ with $3K$ pulses and optimize its $3K+1$ free-evolution intervals directly. The numerical procedure and optimized schedules are provided in the SM. For $K=1$, this procedure recovers XY4~\cite{maudsley1986modified}; for $K>1$, it produces new sequences.}

{For $J\in[10^{-5},10^{-2}]$, we evaluate $\|U(T)-U_0(T)\|$ using the numerically optimized schedules. As shown in Fig.~\ref{fig:Jcomparison}, the curves approach the predicted $\mathcal O(JT^{K+1})$ scaling as $J$ decreases. For $K>1$, the crossover from slope $2$ at short times to slope $K+1$ at larger times agrees with $\mathcal O(JT^{K+1})+\mathcal O(J^2T^2)$ in Eq.~\eqref{target}.}

%Following prior DD studies~\cite{khodjasteh2005fault,west2010high,west2010near,uhrig2010rigorous,xia2011rigorous,ng2011combining,quiroz2013optimized,yi2026faster}, we take $J\in[10^{-2}, 10^{-5}]$. For each $J$, we plot the error $\|U(T)-U_0(T)\|$ versus $T$ using the numerically optimized schedules (see SM). The results are shown in Fig.~\ref{fig:Jcomparison}. As $J$ decreases, the scaling approaches the expected order-$K$ behavior (shaded guides). A crossover from slope $K$ to $2$ is observed, consistent with the prediction $\mathcal{O}(JT^{K+1})+\mathcal{O}(J^2T^2)$ in Eq.~\eqref{target}.

%We next compare our sequence to QDD using the same number of pulses. Specifically, we benchmark our $K=5$ DD against QDD with $K_{\rm QDD}=3$, both of which use $16$ pulses. For the same Hamiltonian model, we sample 100 random product initial states and plot the average trace distance between the reduced system state and the ideal reduced state for various values of $J$. As $J$ decreases, the advantage of our construction becomes more evident, as expected. The improvement is also more visible at larger $T$, where higher-order cancellation in $T$ becomes apparent. (See SM for additional simulations.)

{We next compare our sequence with QDD at equal pulse count.
Our $K=5$ sequence and $K_{\rm QDD}=3$ QDD both use $16$ pulses. Figure~\ref{fig:QDD_comparison_main} shows the trace distance from the ideal reduced state, averaged over $100$ random product initial states. The advantage of our sequence becomes more evident as $J$ decreases and at larger $T$, where its higher-order cancellation becomes visible. Additional comparisons are provided in the SM. We also extend CHaDD to suppress arbitrary $2$-local noise on graphs, obtaining high-order sequences with $L=\mathcal O(\chi K)$. An explicit $3\times3$-grid construction and numerical verification are given in the End Matter.}

%Finally, we consider a $9$-qubit system on a $3\times3$ open grid with single-axis couplings $H_{SB}=\sum_{v\in V} Z_v\otimes B_v+\sum_{(u,v)\in E} Z_uZ_v\otimes B_{uv}$ where $(u,v)$ runs over nearest-neighbor pairs on the grid. For this bipartite graph, CHaDD~\cite{brown2025efficient} achieves first-order decoupling by alternating collective $X$ pulses on sublattices $A$ and $B$. Applying Theorem~\ref{thm:existence} yields $L\le(2\chi-1)K=3K$ for $\chi=2$. The optimized sequence from the single-qubit case can be reused directly under $X\mapsto X_A$ and $Z\mapsto X_B$, giving a $3K(+1)$-pulse sequence that suppresses the leading error to order $K$, consistent with the numerical results in Fig.~\ref{fig:chadd_3x3}. More generally, arbitrary $2$-local noise on a graph of chromatic number $\chi$ admits high-order CHaDD sequences with $L=\mathcal O(\chi K)$, as discussed in the SM.

\textit{Conclusion---}
We introduced a high-order DD construction that suppresses arbitrary system-bath noise using only $\mathcal{O}(|\mathcal G|K)$ pulses by mapping the DD design problem to the necklace-splitting problem. In the weak-coupling regime where the $\mathcal{O}(J)$ term dominates, our sequence matches the error scaling of state-of-the-art DD schemes with fewer pulses, which can yield an exponential reduction in pulse counts for local noise. We also provide explicit sequences of length $3K$ and numerically demonstrate both improved performance over QDD and the realization of higher-order CHaDD.

Our approach provides a practical route to extending qubit coherence times and may serve as a foundation for more resource-efficient pulse-sequence implementations in tasks such as quantum sensing, noise spectroscopy, and quantum simulation. Looking ahead, it would be exciting to implement the deterministic algorithm to generate and benchmark new sequences for large systems. {Another important direction is to make these pulse sequence robust to finite pulse widths and various control imperfections.} More broadly, we hope that our work inspires further connections between quantum control and concrete mathematical problems, leading to more efficient and elegant protocol designs.

\textit{Acknowledgments---}  This material is based upon work supported by the U.S. Department of Energy, Office of Science, Advanced Scientific Computing Research (ASCR), Express: 2023 Exploratory Research For Extreme-scale Science Program under Award Number DE-SC0024685. Additional support by the DOE Early Career Research Award No. DE-SC0026373 is acknowledged.

\textit{Data Availability---} The data that support the findings of this article are openly available \cite{kim2026high}.

\bibliography{ref}

@article{viola1999dynamical,
  title = {Dynamical Decoupling of Open Quantum Systems},
  author = {Viola, Lorenza and Knill, Emanuel and Lloyd, Seth},
  journal = {Phys. Rev. Lett.},
  volume = {82},
  issue = {12},
  pages = {2417--2421},
  numpages = {0},
  year = {1999},
  month = {Mar},
  publisher = {American Physical Society},
  doi = {10.1103/PhysRevLett.82.2417},
  url = {https://link.aps.org/doi/10.1103/PhysRevLett.82.2417}
}

@article{bylander2011noise,
	author = {Bylander, Jonas and Gustavsson, Simon and Yan, Fei and Yoshihara, Fumiki and Harrabi, Khalil and Fitch, George and Cory, David G. and Nakamura, Yasunobu and Tsai, Jaw-Shen and Oliver, William D.},
	da = {2011/07/01},
	doi = {10.1038/nphys1994},
	id = {Bylander2011},
	isbn = {1745-2481},
	journal = {Nat. Phys.},
	number = {7},
	pages = {565--570},
	title = {Noise spectroscopy through dynamical decoupling with a superconducting flux qubit},
	ty = {JOUR},
	url = {https://doi.org/10.1038/nphys1994},
	volume = {7},
	year = {2011}}

@article{xu2012coherence,
  title = {Coherence-Protected Quantum Gate by Continuous Dynamical Decoupling in Diamond},
  author = {Xu, Xiangkun and Wang, Zixiang and Duan, Changkui and Huang, Pu and Wang, Pengfei and Wang, Ya and Xu, Nanyang and Kong, Xi and Shi, Fazhan and Rong, Xing and Du, Jiangfeng},
  journal = {Phys. Rev. Lett.},
  volume = {109},
  issue = {7},
  pages = {070502},
  numpages = {5},
  year = {2012},
  month = {Aug},
  publisher = {American Physical Society},
  doi = {10.1103/PhysRevLett.109.070502},
  url = {https://link.aps.org/doi/10.1103/PhysRevLett.109.070502}
}

@article{pokharel2018demonstration,
  title = {Demonstration of Fidelity Improvement Using Dynamical Decoupling with Superconducting Qubits},
  author = {Pokharel, Bibek and Anand, Namit and Fortman, Benjamin and Lidar, Daniel A.},
  journal = {Phys. Rev. Lett.},
  volume = {121},
  issue = {22},
  pages = {220502},
  numpages = {6},
  year = {2018},
  month = {Nov},
  publisher = {American Physical Society},
  doi = {10.1103/PhysRevLett.121.220502},
  url = {https://link.aps.org/doi/10.1103/PhysRevLett.121.220502}
}

@article{tripathi2022suppression,
  title = {Suppression of Crosstalk in Superconducting Qubits Using Dynamical Decoupling},
  author = {Tripathi, Vinay and Chen, Huo and Khezri, Mostafa and Yip, Ka-Wa and Levenson-Falk, E.M. and Lidar, Daniel A.},
  journal = {Phys. Rev. Appl.},
  volume = {18},
  issue = {2},
  pages = {024068},
  numpages = {24},
  year = {2022},
  month = {Aug},
  publisher = {American Physical Society},
  doi = {10.1103/PhysRevApplied.18.024068},
  url = {https://link.aps.org/doi/10.1103/PhysRevApplied.18.024068}
}

@misc{vezvaee2025surface,
      title={Surface code scaling on heavy-hex superconducting quantum processors}, 
      author={Arian Vezvaee and Cesar Benito and Mario Morford-Oberst and Alejandro Bermudez and Daniel A. Lidar},
      year={2025},
      eprint={2510.18847},
      archivePrefix={arXiv},
      primaryClass={quant-ph},
      url={https://arxiv.org/abs/2510.18847}, 
}

@article{magnus1954on,
	author = {Magnus, Wilhelm},
	doi = {https://doi.org/10.1002/cpa.3160070404},
	journal = {Communications on Pure and Applied Mathematics},
	number = {4},
	pages = {649-673},
	title = {On the exponential solution of differential equations for a linear operator},
	url = {https://onlinelibrary.wiley.com/doi/abs/10.1002/cpa.3160070404},
	volume = {7},
	year = {1954},}

@article{evert2025syncopated,
  title = {Syncopated dynamical decoupling to suppress crosstalk in quantum circuits},
  author = {Evert, Bram and Gonzalez Izquierdo, Zoe and Sud, James and Hu, Hong-Ye and Grabbe, Shon and Rieffel, Eleanor G. and Reagor, Matthew J. and Wang, Zhihui},
  journal = {Phys. Rev. Appl.},
  volume = {24},
  issue = {4},
  pages = {044025},
  numpages = {15},
  year = {2025},
  month = {Oct},
  publisher = {American Physical Society},
  doi = {10.1103/8lxc-lvv1},
  url = {https://link.aps.org/doi/10.1103/8lxc-lvv1}
}

@article{blanes2009the,
   title={The Magnus expansion and some of its applications},
   volume={470},
   ISSN={0370-1573},
   url={http://dx.doi.org/10.1016/j.physrep.2008.11.001},
   DOI={10.1016/j.physrep.2008.11.001},
   number={5–6},
   journal={Physics Reports},
   publisher={Elsevier BV},
   author={Blanes, S. and Casas, F. and Oteo, J.A. and Ros, J.},
   year={2009},
   month=jan, pages={151–238} }

@article{ezzell2023dynamical,
  title = {Dynamical decoupling for superconducting qubits: A performance survey},
  author = {Ezzell, Nic and Pokharel, Bibek and Tewala, Lina and Quiroz, Gregory and Lidar, Daniel A.},
  journal = {Phys. Rev. Appl.},
  volume = {20},
  issue = {6},
  pages = {064027},
  numpages = {42},
  year = {2023},
  month = {Dec},
  publisher = {American Physical Society},
  doi = {10.1103/PhysRevApplied.20.064027},
  url = {https://link.aps.org/doi/10.1103/PhysRevApplied.20.064027}
}

@article{farfrunik2015optimizing,
  title = {Optimizing a dynamical decoupling protocol for solid-state electronic spin ensembles in diamond},
  author = {Farfurnik, D. and Jarmola, A. and Pham, L. M. and Wang, Z. H. and Dobrovitski, V. V. and Walsworth, R. L. and Budker, D. and Bar-Gill, N.},
  journal = {Phys. Rev. B},
  volume = {92},
  issue = {6},
  pages = {060301},
  numpages = {5},
  year = {2015},
  month = {Aug},
  publisher = {American Physical Society},
  doi = {10.1103/PhysRevB.92.060301},
  url = {https://link.aps.org/doi/10.1103/PhysRevB.92.060301}
}

@article{biercuk2009optimized,
	author = {Biercuk, Michael J. and Uys, Hermann and VanDevender, Aaron P. and Shiga, Nobuyasu and Itano, Wayne M. and Bollinger, John J.},
	da = {2009/04/01},
	doi = {10.1038/nature07951},
	id = {Biercuk2009},
	isbn = {1476-4687},
	journal = {Nature},
	number = {7241},
	pages = {996--1000},
	title = {Optimized dynamical decoupling in a model quantum memory},
	ty = {JOUR},
	url = {https://doi.org/10.1038/nature07951},
	volume = {458},
	year = {2009}}

@article{
lange2010universal,
author = {G. de Lange  and Z. H. Wang  and D. Ristè  and V. V. Dobrovitski  and R. Hanson },
title = {Universal Dynamical Decoupling of a Single Solid-State Spin from a Spin Bath},
journal = {Science},
volume = {330},
number = {6000},
pages = {60-63},
year = {2010},
doi = {10.1126/science.1192739}}

@article{acharya2024quantum,
	author = {Acharya, Rajeev and Abanin, Dmitry A. and Aghababaie-Beni, Laleh and Aleiner, Igor and Andersen, Trond I. and others},
	da = {2025/02/01},
	doi = {10.1038/s41586-024-08449-y},
	id = {Acharya2025},
	isbn = {1476-4687},
	journal = {Nature},
	number = {8052},
	pages = {920--926},
	title = {Quantum error correction below the surface code threshold},
	ty = {JOUR},
	url = {https://doi.org/10.1038/s41586-024-08449-y},
	volume = {638},
	year = {2025}}

@article{kim2023scalable,
	author = {Kim, Youngseok and Wood, Christopher J. and Yoder, Theodore J. and Merkel, Seth T. and Gambetta, Jay M. and Temme, Kristan and Kandala, Abhinav},
	da = {2023/05/01},
	doi = {10.1038/s41567-022-01914-3},
	id = {Kim2023},
	isbn = {1745-2481},
	journal = {Nat. Phys.},
	number = {5},
	pages = {752--759},
	title = {Scalable error mitigation for noisy quantum circuits produces competitive expectation values},
	ty = {JOUR},
	url = {https://doi.org/10.1038/s41567-022-01914-3},
	volume = {19},
	year = {2023}}

@misc{paetznick2024demonstration,
      title={Demonstration of logical qubits and repeated error correction with better-than-physical error rates}, 
      author={A. Paetznick and M. P. da Silva and C. Ryan-Anderson and J. M. Bello-Rivas and J. P. Campora III and A. Chernoguzov and J. M. Dreiling and C. Foltz and F. Frachon and J. P. Gaebler and T. M. Gatterman and L. Grans-Samuelsson and D. Gresh and D. Hayes and N. Hewitt and C. Holliman and C. V. Horst and J. Johansen and D. Lucchetti and Y. Matsuoka and M. Mills and S. A. Moses and B. Neyenhuis and A. Paz and J. Pino and P. Siegfried and A. Sundaram and D. Tom and S. J. Wernli and M. Zanner and R. P. Stutz and K. M. Svore},
      year={2024},
      eprint={2404.02280},
      archivePrefix={arXiv},
      primaryClass={quant-ph},
      url={https://arxiv.org/abs/2404.02280}, 
}

@article{uhrig2007keeping,
  title = {Keeping a Quantum Bit Alive by Optimized $\ensuremath{\pi}$-Pulse Sequences},
  author = {Uhrig, G\"otz S.},
  journal = {Phys. Rev. Lett.},
  volume = {98},
  issue = {10},
  pages = {100504},
  numpages = {4},
  year = {2007},
  month = {Mar},
  publisher = {American Physical Society},
  doi = {10.1103/PhysRevLett.98.100504},
  url = {https://link.aps.org/doi/10.1103/PhysRevLett.98.100504}
}

@article{khodjasteh2005fault,
  title = {Fault-Tolerant Quantum Dynamical Decoupling},
  author = {Khodjasteh, K. and Lidar, D. A.},
  journal = {Phys. Rev. Lett.},
  volume = {95},
  issue = {18},
  pages = {180501},
  numpages = {4},
  year = {2005},
  month = {Oct},
  publisher = {American Physical Society},
  doi = {10.1103/PhysRevLett.95.180501},
  url = {https://link.aps.org/doi/10.1103/PhysRevLett.95.180501}
}

@article{bluvstein2025a,
   title={A fault-tolerant neutral-atom architecture for universal quantum computation},
   volume={649},
   ISSN={1476-4687},
   url={http://dx.doi.org/10.1038/s41586-025-09848-5},
   DOI={10.1038/s41586-025-09848-5},
   number={8095},
   journal={Nature},
   publisher={Springer Science and Business Media LLC},
   author={Bluvstein, Dolev and Geim, Alexandra A. and Li, Sophie H. and Evered, Simon J. and Bonilla Ataides, J. Pablo and Baranes, Gefen and Gu, Andi and Manovitz, Tom and Xu, Muqing and Kalinowski, Marcin and Majidy, Shayan and Kokail, Christian and Maskara, Nishad and Trapp, Elias C. and Stewart, Luke M. and Hollerith, Simon and Zhou, Hengyun and Gullans, Michael J. and Yelin, Susanne F. and Greiner, Markus and Vuletić, Vladan and Cain, Madelyn and Lukin, Mikhail D.},
   year={2025},
   month=Nov, pages={39–46} }

@article{viola2003robust,
  title = {Robust Dynamical Decoupling of Quantum Systems with Bounded Controls},
  author = {Viola, Lorenza and Knill, Emanuel},
  journal = {Phys. Rev. Lett.},
  volume = {90},
  issue = {3},
  pages = {037901},
  numpages = {4},
  year = {2003},
  month = {Jan},
  publisher = {American Physical Society},
  doi = {10.1103/PhysRevLett.90.037901},
  url = {https://link.aps.org/doi/10.1103/PhysRevLett.90.037901}
}

@article{zanardi1999symmetrizing,
   title={Symmetrizing evolutions},
   volume={258},
   ISSN={0375-9601},
   url={http://dx.doi.org/10.1016/S0375-9601(99)00365-5},
   DOI={10.1016/s0375-9601(99)00365-5},
   number={2–3},
   journal={Physics Letters A},
   publisher={Elsevier BV},
   author={Zanardi, Paolo},
   year={1999},
   month=jul, pages={77–82} }

@article{yi2026faster,
  title = {Faster Randomized Dynamical Decoupling},
  author = {Yi, Changhao and Kim, Leeseok and Marvian, Milad},
  journal = {Phys. Rev. Lett.},
  volume = {136},
  issue = {1},
  pages = {010601},
  numpages = {7},
  year = {2026},
  month = {Jan},
  publisher = {American Physical Society},
  doi = {10.1103/fk7j-y1vl},
  url = {https://link.aps.org/doi/10.1103/fk7j-y1vl}
}

@article{west2010near,
  title = {Near-Optimal Dynamical Decoupling of a Qubit},
  author = {West, Jacob R. and Fong, Bryan H. and Lidar, Daniel A.},
  journal = {Phys. Rev. Lett.},
  volume = {104},
  issue = {13},
  pages = {130501},
  numpages = {4},
  year = {2010},
  month = {Apr},
  publisher = {American Physical Society},
  doi = {10.1103/PhysRevLett.104.130501},
  url = {https://link.aps.org/doi/10.1103/PhysRevLett.104.130501}
}

@article{uhrig2010rigorous,
  title = {Rigorous bounds for optimal dynamical decoupling},
  author = {Uhrig, G\"otz S. and Lidar, Daniel A.},
  journal = {Phys. Rev. A},
  volume = {82},
  issue = {1},
  pages = {012301},
  numpages = {9},
  year = {2010},
  month = {Jul},
  publisher = {American Physical Society},
  doi = {10.1103/PhysRevA.82.012301},
  url = {https://link.aps.org/doi/10.1103/PhysRevA.82.012301}
}

@article{xia2011rigorous,
  title = {Rigorous performance bounds for quadratic and nested dynamical decoupling},
  author = {Xia, Yuhou and Uhrig, G\"otz S. and Lidar, Daniel A.},
  journal = {Phys. Rev. A},
  volume = {84},
  issue = {6},
  pages = {062332},
  numpages = {11},
  year = {2011},
  month = {Dec},
  publisher = {American Physical Society},
  doi = {10.1103/PhysRevA.84.062332},
  url = {https://link.aps.org/doi/10.1103/PhysRevA.84.062332}
}

@article{yang2008universality,
  title = {Universality of Uhrig Dynamical Decoupling for Suppressing Qubit Pure Dephasing and Relaxation},
  author = {Yang, Wen and Liu, Ren-Bao},
  journal = {Phys. Rev. Lett.},
  volume = {101},
  issue = {18},
  pages = {180403},
  numpages = {4},
  year = {2008},
  month = {Oct},
  publisher = {American Physical Society},
  doi = {10.1103/PhysRevLett.101.180403},
  url = {https://link.aps.org/doi/10.1103/PhysRevLett.101.180403}
}

@article{wang2011protection,
  title = {Protection of quantum systems by nested dynamical decoupling},
  author = {Wang, Zhen-Yu and Liu, Ren-Bao},
  journal = {Phys. Rev. A},
  volume = {83},
  issue = {2},
  pages = {022306},
  numpages = {10},
  year = {2011},
  month = {Feb},
  publisher = {American Physical Society},
  doi = {10.1103/PhysRevA.83.022306},
  url = {https://link.aps.org/doi/10.1103/PhysRevA.83.022306}
}

@misc{filosratsikas2018hardness,
      title={Hardness Results for Consensus-Halving}, 
      author={Aris Filos-Ratsikas and Soren Kristoffer Stiil Frederiksen and Paul W. Goldberg and Jie Zhang},
      year={2018},
      eprint={1609.05136},
      archivePrefix={arXiv},
      primaryClass={cs.GT},
      url={https://arxiv.org/abs/1609.05136}, 
}

@article{jiang2011universal,
  title = {Universal dynamical decoupling of multiqubit states from environment},
  author = {Jiang, Liang and Imambekov, Adilet},
  journal = {Phys. Rev. A},
  volume = {84},
  issue = {6},
  pages = {060302},
  numpages = {4},
  year = {2011},
  month = {Dec},
  publisher = {American Physical Society},
  doi = {10.1103/PhysRevA.84.060302},
  url = {https://link.aps.org/doi/10.1103/PhysRevA.84.060302}
}

@article{khodjasteh2007performance,
  title = {Performance of deterministic dynamical decoupling schemes: Concatenated and periodic pulse sequences},
  author = {Khodjasteh, Kaveh and Lidar, Daniel A.},
  journal = {Phys. Rev. A},
  volume = {75},
  issue = {6},
  pages = {062310},
  numpages = {16},
  year = {2007},
  month = {Jun},
  publisher = {American Physical Society},
  doi = {10.1103/PhysRevA.75.062310},
  url = {https://link.aps.org/doi/10.1103/PhysRevA.75.062310}
}

@article{quiroz2013optimized,
   title={Optimized dynamical decoupling via genetic algorithms},
   volume={88},
   ISSN={1094-1622},
   url={http://dx.doi.org/10.1103/PhysRevA.88.052306},
   DOI={10.1103/physreva.88.052306},
   number={5},
   journal={Physical Review A},
   publisher={American Physical Society (APS)},
   author={Quiroz, Gregory and Lidar, Daniel A.},
   year={2013},
   month=nov }

@article{hahn1950spin,
  title = {Spin Echoes},
  author = {Hahn, E. L.},
  journal = {Phys. Rev.},
  volume = {80},
  issue = {4},
  pages = {580--594},
  numpages = {0},
  year = {1950},
  month = {Nov},
  publisher = {American Physical Society},
  doi = {10.1103/PhysRev.80.580},
  url = {https://link.aps.org/doi/10.1103/PhysRev.80.580}
}

@article{meiboom1958modified,
    author = {Meiboom, S. and Gill, D.},
    title = "{Modified Spin‐Echo Method for Measuring Nuclear Relaxation Times}",
    journal = {Rev. Sci. Instrum.},
    volume = {29},
    number = {8},
    pages = {688-691},
    year = {1958},
    month = {08},
    issn = {0034-6748},
    doi = {10.1063/1.1716296},
    url = {https://doi.org/10.1063/1.1716296},}

@article{maudsley1986modified,
	author = {Maudsley, A. A},
	journal = {Journal of Magnetic Resonance},
	month = oct,
	number = {3},
	pages = {488--491},
	title = {{Modified Carr-Purcell-Meiboom-Gill sequence for NMR Fourier imaging applications}},
	url = {http://www.sciencedirect.com/science/article/pii/0022236486901605},
	urldate = {2020-06-19},
	volume = {69},
	year = {1986}}

@article{ng2011combining,
  title = {Combining dynamical decoupling with fault-tolerant quantum computation},
  author = {Ng, Hui Khoon and Lidar, Daniel A. and Preskill, John},
  journal = {Phys. Rev. A},
  volume = {84},
  issue = {1},
  pages = {012305},
  numpages = {38},
  year = {2011},
  month = {Jul},
  publisher = {American Physical Society},
  doi = {10.1103/PhysRevA.84.012305},
  url = {https://link.aps.org/doi/10.1103/PhysRevA.84.012305}
}

@article{viola2005random,
  title = {Random Decoupling Schemes for Quantum Dynamical Control and Error Suppression},
  author = {Viola, Lorenza and Knill, Emanuel},
  journal = {Phys. Rev. Lett.},
  volume = {94},
  issue = {6},
  pages = {060502},
  numpages = {4},
  year = {2005},
  month = {Feb},
  publisher = {American Physical Society},
  doi = {10.1103/PhysRevLett.94.060502},
  url = {https://link.aps.org/doi/10.1103/PhysRevLett.94.060502}
}

@article{genov2017arbitrarily,
  title = {Arbitrarily Accurate Pulse Sequences for Robust Dynamical Decoupling},
  author = {Genov, Genko T. and Schraft, Daniel and Vitanov, Nikolay V. and Halfmann, Thomas},
  journal = {Phys. Rev. Lett.},
  volume = {118},
  issue = {13},
  pages = {133202},
  numpages = {5},
  year = {2017},
  month = {Mar},
  publisher = {American Physical Society},
  doi = {10.1103/PhysRevLett.118.133202},
  url = {https://link.aps.org/doi/10.1103/PhysRevLett.118.133202}
}

@article{alon1987splitting,
	author = {Noga Alon},
	doi = {https://doi.org/10.1016/0001-8708(87)90055-7},
	issn = {0001-8708},
	journal = {Advances in Mathematics},
	number = {3},
	pages = {247-253},
	title = {Splitting necklaces},
	url = {https://www.sciencedirect.com/science/article/pii/0001870887900557},
	volume = {63},
	year = {1987}}

@InProceedings{alon2021efficient,
  author =	{Alon, Noga and Graur, Andrei},
  title =	{{Efficient Splitting of Necklaces}},
  booktitle =	{48th International Colloquium on Automata, Languages, and Programming (ICALP 2021)},
  pages =	{14:1--14:17},
  series =	{Leibniz International Proceedings in Informatics (LIPIcs)},
  ISBN =	{978-3-95977-195-5},
  ISSN =	{1868-8969},
  year =	{2021},
  volume =	{198},
  editor =	{Bansal, Nikhil and Merelli, Emanuela and Worrell, James},
  publisher =	{Schloss Dagstuhl -- Leibniz-Zentrum f{\"u}r Informatik},
  address =	{Dagstuhl, Germany},
  URL =		{https://drops.dagstuhl.de/entities/document/10.4230/LIPIcs.ICALP.2021.14},
  URN =		{urn:nbn:de:0030-drops-140832},
  doi =		{10.4230/LIPIcs.ICALP.2021.14}
}

@article{uhrig2009concatenated,
  title = {Concatenated Control Sequences Based on Optimized Dynamic Decoupling},
  author = {Uhrig, G\"otz S.},
  journal = {Phys. Rev. Lett.},
  volume = {102},
  issue = {12},
  pages = {120502},
  numpages = {4},
  year = {2009},
  month = {Mar},
  publisher = {American Physical Society},
  doi = {10.1103/PhysRevLett.102.120502},
  url = {https://link.aps.org/doi/10.1103/PhysRevLett.102.120502}
}

@article{alvarez2010performance,
  title = {Performance comparison of dynamical decoupling sequences for a qubit in a rapidly fluctuating spin bath},
  author = {\'Alvarez, Gonzalo A. and Ajoy, Ashok and Peng, Xinhua and Suter, Dieter},
  journal = {Phys. Rev. A},
  volume = {82},
  issue = {4},
  pages = {042306},
  numpages = {13},
  year = {2010},
  month = {Oct},
  publisher = {American Physical Society},
  doi = {10.1103/PhysRevA.82.042306},
  url = {https://link.aps.org/doi/10.1103/PhysRevA.82.042306}
}

@article{ajoy2011optimal,
  title = {Optimal pulse spacing for dynamical decoupling in the presence of a purely dephasing spin bath},
  author = {Ajoy, Ashok and \'Alvarez, Gonzalo A. and Suter, Dieter},
  journal = {Phys. Rev. A},
  volume = {83},
  issue = {3},
  pages = {032303},
  numpages = {14},
  year = {2011},
  month = {Mar},
  publisher = {American Physical Society},
  doi = {10.1103/PhysRevA.83.032303},
  url = {https://link.aps.org/doi/10.1103/PhysRevA.83.032303}
}

@article{zhao2012decoherence,
  title = {Decoherence and dynamical decoupling control of nitrogen vacancy center electron spins in nuclear spin baths},
  author = {Zhao, Nan and Ho, Sai-Wah and Liu, Ren-Bao},
  journal = {Phys. Rev. B},
  volume = {85},
  issue = {11},
  pages = {115303},
  numpages = {18},
  year = {2012},
  month = {Mar},
  publisher = {American Physical Society},
  doi = {10.1103/PhysRevB.85.115303},
  url = {https://link.aps.org/doi/10.1103/PhysRevB.85.115303}
}

@article{west2010high,
  title = {High Fidelity Quantum Gates via Dynamical Decoupling},
  author = {West, Jacob R. and Lidar, Daniel A. and Fong, Bryan H. and Gyure, Mark F.},
  journal = {Phys. Rev. Lett.},
  volume = {105},
  issue = {23},
  pages = {230503},
  numpages = {4},
  year = {2010},
  month = {Dec},
  publisher = {American Physical Society},
  doi = {10.1103/PhysRevLett.105.230503},
  url = {https://link.aps.org/doi/10.1103/PhysRevLett.105.230503}
}

@article{witzel2007concatenated,
  title = {Concatenated dynamical decoupling in a solid-state spin bath},
  author = {Witzel, W. M. and Das Sarma, S.},
  journal = {Phys. Rev. B},
  volume = {76},
  issue = {24},
  pages = {241303},
  numpages = {4},
  year = {2007},
  month = {Dec},
  publisher = {American Physical Society},
  doi = {10.1103/PhysRevB.76.241303},
  url = {https://link.aps.org/doi/10.1103/PhysRevB.76.241303}
}

@misc{suppmat,
  note = {See Supplemental Material at [URL] for proofs of the main theorem and lower bound, extensions to classical and bosonic noise, optimized pulse schedules, additional numerical results, analyses of pulse imperfections, and high-order CHaDD constructions, which includes Refs.~[73–77].}
}

@article{biercuk2009experimental,
  title = {Experimental Uhrig dynamical decoupling using trapped ions},
  author = {Biercuk, Michael J. and Uys, Hermann and VanDevender, Aaron P. and Shiga, Nobuyasu and Itano, Wayne M. and Bollinger, John J.},
  journal = {Phys. Rev. A},
  volume = {79},
  issue = {6},
  pages = {062324},
  numpages = {12},
  year = {2009},
  month = {Jun},
  publisher = {American Physical Society},
  doi = {10.1103/PhysRevA.79.062324},
  url = {https://link.aps.org/doi/10.1103/PhysRevA.79.062324}
}

@misc{kasatkin2026quantum,
      title={Quantum Error Correction and Dynamical Decoupling: Better Together or Apart?}, 
      author={Victor Kasatkin and Mario Morford-Oberst and Arian Vezvaee and Daniel A. Lidar},
      year={2026},
      eprint={2602.19042},
      archivePrefix={arXiv},
      primaryClass={quant-ph},
      url={https://arxiv.org/abs/2602.19042}, 
}

@article{paz2014general,
  title = {General Transfer-Function Approach to Noise Filtering in Open-Loop Quantum Control},
  author = {Paz-Silva, Gerardo A. and Viola, Lorenza},
  journal = {Phys. Rev. Lett.},
  volume = {113},
  issue = {25},
  pages = {250501},
  numpages = {5},
  year = {2014},
  month = {Dec},
  publisher = {American Physical Society},
  doi = {10.1103/PhysRevLett.113.250501},
  url = {https://link.aps.org/doi/10.1103/PhysRevLett.113.250501}
}

@article{clausen2012task,
  title = {Task-optimized control of open quantum systems},
  author = {Clausen, Jens and Bensky, Guy and Kurizki, Gershon},
  journal = {Phys. Rev. A},
  volume = {85},
  issue = {5},
  pages = {052105},
  numpages = {10},
  year = {2012},
  month = {May},
  publisher = {American Physical Society},
  doi = {10.1103/PhysRevA.85.052105},
  url = {https://link.aps.org/doi/10.1103/PhysRevA.85.052105}
}

@article{gordon2007universal,
  title = {Universal dephasing control during quantum computation},
  author = {Gordon, Goren and Kurizki, Gershon},
  journal = {Phys. Rev. A},
  volume = {76},
  issue = {4},
  pages = {042310},
  numpages = {5},
  year = {2007},
  month = {Oct},
  publisher = {American Physical Society},
  doi = {10.1103/PhysRevA.76.042310},
  url = {https://link.aps.org/doi/10.1103/PhysRevA.76.042310}
}

@article{kofman2004unified,
  title = {Unified Theory of Dynamically Suppressed Qubit Decoherence in Thermal Baths},
  author = {Kofman, A. G. and Kurizki, G.},
  journal = {Phys. Rev. Lett.},
  volume = {93},
  issue = {13},
  pages = {130406},
  numpages = {4},
  year = {2004},
  month = {Sep},
  publisher = {American Physical Society},
  doi = {10.1103/PhysRevLett.93.130406},
  url = {https://link.aps.org/doi/10.1103/PhysRevLett.93.130406}
}

@article{uhrig2008exact,
	author = {Uhrig, G\"otz S.},
	doi = {10.1088/1367-2630/10/8/083024},
	journal = {New Journal of Physics},
	month = {aug},
	number = {8},
	pages = {083024},
	title = {Exact results on dynamical decoupling by $\ensuremath{\pi}$ pulses in quantum information processes},
	url = {https://doi.org/10.1088/1367-2630/10/8/083024},
	volume = {10},
	year = {2008}}

@article{cywinski2008how,
  title = {How to enhance dephasing time in superconducting qubits},
  author = {Cywi\ifmmode \acute{n}\else \'{n}\fi{}ski, \L{}ukasz and Lutchyn, Roman M. and Nave, Cody P. and Das Sarma, S.},
  journal = {Phys. Rev. B},
  volume = {77},
  issue = {17},
  pages = {174509},
  numpages = {11},
  year = {2008},
  month = {May},
  publisher = {American Physical Society},
  doi = {10.1103/PhysRevB.77.174509},
  url = {https://link.aps.org/doi/10.1103/PhysRevB.77.174509}
}

@article{gross2024characterizing,
	author = {Gross, Jonathan A. and Genois, {\'E}lie and Debroy, Dripto M. and Zhang, Yaxing and Mruczkiewicz, Wojciech and Cian, Ze-Pei and Jiang, Zhang},
	da = {2024/11/22},
	doi = {10.1038/s41534-024-00917-7},
	id = {Gross2024},
	isbn = {2056-6387},
	journal = {npj Quantum Information},
	number = {1},
	pages = {123},
	title = {Characterizing coherent errors using matrix-element amplification},
	ty = {JOUR},
	url = {https://doi.org/10.1038/s41534-024-00917-7},
	volume = {10},
	year = {2024}}

@article{souza2012robust,
   title={Robust dynamical decoupling},
   volume={370},
   ISSN={1471-2962},
   url={http://dx.doi.org/10.1098/rsta.2011.0355},
   DOI={10.1098/rsta.2011.0355},
   number={1976},
   journal={Philosophical Transactions of the Royal Society A: Mathematical, Physical and Engineering Sciences},
   publisher={The Royal Society},
   author={Souza, Alexandre M. and Álvarez, Gonzalo A. and Suter, Dieter},
   year={2012},
   month=oct, pages={4748–4769} }

@article{brown2025efficient,
  title = {Efficient Chromatic-Number-Based Multiqubit Decoherence and Crosstalk Suppression},
  author = {Brown, Amy F. and Lidar, Daniel A.},
  journal = {PRX Quantum},
  volume = {6},
  issue = {2},
  pages = {020354},
  numpages = {15},
  year = {2025},
  month = {Jun},
  publisher = {American Physical Society},
  doi = {10.1103/1d4l-73x6},
  url = {https://link.aps.org/doi/10.1103/1d4l-73x6}
}

@article{barthel2010interlaced,
  title = {Interlaced Dynamical Decoupling and Coherent Operation of a Singlet-Triplet Qubit},
  author = {Barthel, C. and Medford, J. and Marcus, C. M. and Hanson, M. P. and Gossard, A. C.},
  journal = {Phys. Rev. Lett.},
  volume = {105},
  issue = {26},
  pages = {266808},
  numpages = {4},
  year = {2010},
  month = {Dec},
  publisher = {American Physical Society},
  doi = {10.1103/PhysRevLett.105.266808},
  url = {https://link.aps.org/doi/10.1103/PhysRevLett.105.266808}
}

@article{ryan2010robust,
  title = {Robust Decoupling Techniques to Extend Quantum Coherence in Diamond},
  author = {Ryan, C. A. and Hodges, J. S. and Cory, D. G.},
  journal = {Phys. Rev. Lett.},
  volume = {105},
  issue = {20},
  pages = {200402},
  numpages = {4},
  year = {2010},
  month = {Nov},
  publisher = {American Physical Society},
  doi = {10.1103/PhysRevLett.105.200402},
  url = {https://link.aps.org/doi/10.1103/PhysRevLett.105.200402}
}

@article{medford2012scaling,
  title = {Scaling of Dynamical Decoupling for Spin Qubits},
  author = {Medford, J. and Cywi\ifmmode \acute{n}\else \'{n}\fi{}ski, \L{}. and Barthel, C. and Marcus, C. M. and Hanson, M. P. and Gossard, A. C.},
  journal = {Phys. Rev. Lett.},
  volume = {108},
  issue = {8},
  pages = {086802},
  numpages = {5},
  year = {2012},
  month = {Feb},
  publisher = {American Physical Society},
  doi = {10.1103/PhysRevLett.108.086802},
  url = {https://link.aps.org/doi/10.1103/PhysRevLett.108.086802}
}

@article{souza2011robust,
  title = {Robust Dynamical Decoupling for Quantum Computing and Quantum Memory},
  author = {Souza, Alexandre M. and \'Alvarez, Gonzalo A. and Suter, Dieter},
  journal = {Phys. Rev. Lett.},
  volume = {106},
  issue = {24},
  pages = {240501},
  numpages = {4},
  year = {2011},
  month = {Jun},
  publisher = {American Physical Society},
  doi = {10.1103/PhysRevLett.106.240501},
  url = {https://link.aps.org/doi/10.1103/PhysRevLett.106.240501}
}

@article{bookatz2016improved,
   title={Improved Bounded-Strength Decoupling Schemes for Local Hamiltonians},
   volume={62},
   ISSN={1557-9654},
   url={http://dx.doi.org/10.1109/TIT.2016.2535183},
   DOI={10.1109/tit.2016.2535183},
   number={5},
   journal={IEEE Transactions on Information Theory},
   publisher={Institute of Electrical and Electronics Engineers (IEEE)},
   author={Bookatz, Adam D. and Roetteler, Martin and Wocjan, Pawel},
   year={2016},
   month=may, pages={2881–2894} }

@article{terhal2005fault,
  title = {Fault-tolerant quantum computation for local non-Markovian noise},
  author = {Terhal, Barbara M. and Burkard, Guido},
  journal = {Phys. Rev. A},
  volume = {71},
  issue = {1},
  pages = {012336},
  numpages = {11},
  year = {2005},
  month = {Jan},
  publisher = {American Physical Society},
  doi = {10.1103/PhysRevA.71.012336},
  url = {https://link.aps.org/doi/10.1103/PhysRevA.71.012336}
}

@article{aharonov2006fault,
  title = {Fault-Tolerant Quantum Computation with Long-Range Correlated Noise},
  author = {Aharonov, Dorit and Kitaev, Alexei and Preskill, John},
  journal = {Phys. Rev. Lett.},
  volume = {96},
  issue = {5},
  pages = {050504},
  numpages = {4},
  year = {2006},
  month = {Feb},
  publisher = {American Physical Society},
  doi = {10.1103/PhysRevLett.96.050504},
  url = {https://link.aps.org/doi/10.1103/PhysRevLett.96.050504}
}

@article{von2020two,
  title = {Two-Qubit Spectroscopy of Spatiotemporally Correlated Quantum Noise in Superconducting Qubits},
  author = {von L\"upke, Uwe and Beaudoin, F\'elix and Norris, Leigh M. and Sung, Youngkyu and Winik, Roni and Qiu, Jack Y. and Kjaergaard, Morten and Kim, David and Yoder, Jonilyn and Gustavsson, Simon and Viola, Lorenza and Oliver, William D.},
  journal = {PRX Quantum},
  volume = {1},
  issue = {1},
  pages = {010305},
  numpages = {23},
  year = {2020},
  month = {Sep},
  publisher = {American Physical Society},
  doi = {10.1103/PRXQuantum.1.010305},
  url = {https://link.aps.org/doi/10.1103/PRXQuantum.1.010305}
}

@article{degen2017quantum,
  title = {Quantum sensing},
  author = {Degen, C. L. and Reinhard, F. and Cappellaro, P.},
  journal = {Rev. Mod. Phys.},
  volume = {89},
  issue = {3},
  pages = {035002},
  numpages = {39},
  year = {2017},
  month = {Jul},
  publisher = {American Physical Society},
  doi = {10.1103/RevModPhys.89.035002},
  url = {https://link.aps.org/doi/10.1103/RevModPhys.89.035002}
}

@book{breuer2002theory,
  title={The theory of open quantum systems},
  author={Breuer, Heinz-Peter and Petruccione, Francesco},
  year={2002},
  publisher={OUP Oxford}
}

@article{green2013arbitrary,
   title={Arbitrary quantum control of qubits in the presence of universal noise},
   volume={15},
   ISSN={1367-2630},
   url={http://dx.doi.org/10.1088/1367-2630/15/9/095004},
   DOI={10.1088/1367-2630/15/9/095004},
   number={9},
   journal={New Journal of Physics},
   publisher={IOP Publishing},
   author={Green, Todd J and Sastrawan, Jarrah and Uys, Hermann and Biercuk, Michael J},
   year={2013},
   month=sep, pages={095004} }

@article{lidar2008towards,
  title = {Towards Fault Tolerant Adiabatic Quantum Computation},
  author = {Lidar, Daniel A.},
  journal = {Phys. Rev. Lett.},
  volume = {100},
  issue = {16},
  pages = {160506},
  numpages = {4},
  year = {2008},
  month = {Apr},
  publisher = {American Physical Society},
  doi = {10.1103/PhysRevLett.100.160506},
  url = {https://link.aps.org/doi/10.1103/PhysRevLett.100.160506}
}

@article{quiroz2012high,
  title = {High-fidelity adiabatic quantum computation via dynamical decoupling},
  author = {Quiroz, Gregory and Lidar, Daniel A.},
  journal = {Phys. Rev. A},
  volume = {86},
  issue = {4},
  pages = {042333},
  numpages = {9},
  year = {2012},
  month = {Oct},
  publisher = {American Physical Society},
  doi = {10.1103/PhysRevA.86.042333},
  url = {https://link.aps.org/doi/10.1103/PhysRevA.86.042333}
}

@article{khodjasteh2008rigorous,
  title = {Rigorous bounds on the performance of a hybrid dynamical-decoupling quantum-computing scheme},
  author = {Khodjasteh, Kaveh and Lidar, Daniel A.},
  journal = {Phys. Rev. A},
  volume = {78},
  issue = {1},
  pages = {012355},
  numpages = {14},
  year = {2008},
  month = {Jul},
  publisher = {American Physical Society},
  doi = {10.1103/PhysRevA.78.012355},
  url = {https://link.aps.org/doi/10.1103/PhysRevA.78.012355}
}

@article{hahn2025efficiency,
   title={Efficiency of dynamical decoupling for (almost) any spin–boson model},
   volume={19},
   ISSN={2542-4653},
   url={http://dx.doi.org/10.21468/SciPostPhys.19.2.035},
   DOI={10.21468/scipostphys.19.2.035},
   number={2},
   journal={SciPost Physics},
   publisher={Stichting SciPost},
   author={Hahn, Alexander and Burgarth, Daniel and Lonigro, Davide},
   year={2025},
   month=aug }

@article{zhao2014spin,
   title={Spin-boson model with diagonal and off-diagonal coupling to two independent baths: Ground-state phase transition in the deep sub-Ohmic regime},
   volume={140},
   ISSN={1089-7690},
   url={http://dx.doi.org/10.1063/1.4873351},
   DOI={10.1063/1.4873351},
   number={16},
   journal={The Journal of Chemical Physics},
   publisher={AIP Publishing},
   author={Zhao, Yang and Yao, Yao and Chernyak, Vladimir and Zhao, Yang},
   year={2014},
   month=apr }

@article{viola1999universal,
  title = {Universal Control of Decoupled Quantum Systems},
  author = {Viola, Lorenza and Lloyd, Seth and Knill, Emanuel},
  journal = {Phys. Rev. Lett.},
  volume = {83},
  issue = {23},
  pages = {4888--4891},
  numpages = {0},
  year = {1999},
  month = {Dec},
  publisher = {American Physical Society},
  doi = {10.1103/PhysRevLett.83.4888},
  url = {https://link.aps.org/doi/10.1103/PhysRevLett.83.4888}
}

@article{de2013universal,
  title = {Universal Set of Scalable Dynamically Corrected Gates for Quantum Error Correction with Always-on Qubit Couplings},
  author = {De, Amrit and Pryadko, Leonid P.},
  journal = {Phys. Rev. Lett.},
  volume = {110},
  issue = {7},
  pages = {070503},
  numpages = {5},
  year = {2013},
  month = {Feb},
  publisher = {American Physical Society},
  doi = {10.1103/PhysRevLett.110.070503},
  url = {https://link.aps.org/doi/10.1103/PhysRevLett.110.070503}
}

@article{ramon2022qubit,
  title = {Qubit decoherence under two-axis coupling to low-frequency noises},
  author = {Ramon, Guy and Cywi\ifmmode \acute{n}\else \'{n}\fi{}ski, \L{}ukasz},
  journal = {Phys. Rev. B},
  volume = {105},
  issue = {4},
  pages = {L041303},
  numpages = {6},
  year = {2022},
  month = {Jan},
  publisher = {American Physical Society},
  doi = {10.1103/PhysRevB.105.L041303},
  url = {https://link.aps.org/doi/10.1103/PhysRevB.105.L041303}
}

@misc{kim2026high,
  author       = {L. Kim and M. Marvian},
  title        = {High-order dynamical decoupling for arbitrary noise in the weak-coupling regime},
  year         = {2026},
  howpublished = {\url{https://github.com/Leeseok-628/high-order-dd-weak-coupling}},
  note         = {GitHub repository}
}

%%%%%%%%%%%%%%%%%%%%%%%%%%%%%%%%%%%%%%%%%%%%%%%%%%%%%%%%%%%%%%%%%%%%%%%%%%%%%%%
% -------------  END MATTER  -------------
%%%%%%%%%%%%%%%%%%%%%%%%%%%%%%%%%%%%%%%%%%%%%%%%%%%%%%%%%%%%%%%%%%%%%%%%%%%%%%%

\section*{End Matter}

\subsection*{Moment cancellation and the error bound}
\label{sec:moment-error-bound}

To separate the control-induced modulation from the $H_0$ evolution, define
\begin{align}
\widehat H(\tau) &:= \bigl(U_c^\dagger(T\tau)\otimes I_B\bigr) H_{SB} \bigl(U_c(T\tau)\otimes I_B\bigr) \\
&= \sum_\alpha y_\alpha(\tau)\sigma_\alpha\otimes B_\alpha.
\label{eq:control-toggled-interaction}
\end{align}
Unitary invariance gives $\|\widehat H(\tau)\|=\|H_{SB}\|=J$, and Eq.~\eqref{toggling-U0} becomes $\widetilde H(T\tau)=e^{iH_0T\tau} \widehat H(\tau)e^{-iH_0T\tau}$. Moreover, the constraints in Eq.~\eqref{moment-constraints-compact} imply
\begin{align}
\int_0^1\tau^m\widehat H(\tau)d\tau = \sum_\alpha M_{\alpha,m} \bigl(\sigma_\alpha\otimes B_\alpha\bigr)
=0, \nonumber\\
m=0,\ldots,K-1.
\label{eq:operator-moment-identities}
\end{align}

Assume that the Magnus expansion of the time-ordered exponential in Eq.~\eqref{U_interaction} converges; a sufficient condition is $JT<\pi$~\cite{magnus1954on,blanes2009the}. Writing $\mathcal L(X):=[H_0,X]$, the first Magnus term is
\begin{align}
\Omega_1(T) &=-i\int_0^T\widetilde H(t)dt
\\
&=-iT\int_0^1 e^{iT\tau\mathcal L}\widehat H(\tau)d\tau.
\label{eq:first-magnus-term}
\end{align}
Introduce the Taylor remainder through $e^{iz}=\sum_{m=0}^{K-1}(iz)^m/m!+\mathcal R_K(z)$. Because $\mathcal L$ is independent of $\tau$, every polynomial term in Eq.~\eqref{eq:first-magnus-term} is proportional to $\mathcal L^m\left[\int_0^1\tau^m\widehat H(\tau)d\tau\right]$ and therefore vanishes by Eq.~\eqref{eq:operator-moment-identities}. Hence
\begin{align}
\Omega_1(T) = -iT\int_0^1
\mathcal R_K(T\tau\mathcal L)\widehat H(\tau)d\tau.
\label{eq:first-magnus-remainder}
\end{align}
Since $\|H_0\|=1$, we have $\|\mathcal L(X)\|=\|[H_0,X]\|\le2\|X\|$. Therefore,
\begin{align}
\|\mathcal R_K(T\tau\mathcal L)\|
&\le \sum_{m=K}^{\infty}\frac{(2T\tau)^m}{m!} \le \frac{e^{2T}(2T\tau)^K}{K!},
\end{align}
which gives
\begin{align}
\|\Omega_1(T)\| &\le TJ\int_0^1 \frac{e^{2T}(2T\tau)^K}{K!}d\tau
\\
&= \frac{2^Ke^{2T}}{(K+1)!}JT^{K+1} = \mathcal O(JT^{K+1}).
\label{eq:first-magnus-bound}
\end{align}

The second Magnus term obeys
\begin{align}
\|\Omega_2(T)\| &\le \frac12\int_0^Tdt_1\int_0^{t_1}dt_2 \|[\widetilde H(t_1),\widetilde H(t_2)]\| \\ &\le \frac12J^2T^2 = \mathcal O(J^2T^2),
\label{eq:second-magnus-bound}
\end{align}
where we used $\|\widetilde H(t)\|=J$ and $\|[A,B]\|\le2\|A\|\|B\|$. All higher Magnus terms are of order $\mathcal O((JT)^r)$ with $r\ge3$. Consequently,
\begin{align}
\left\|\exp\left[\sum_{r\ge1}\Omega_r(T)\right]-I \right\| = \mathcal O(JT^{K+1})+\mathcal O(J^2T^2).
\end{align}
Finally, Eq.~\eqref{U_interaction} and the unitarity of $U_0(T)$ give
\begin{align}
\|U(T)-U_0(T)\| = \mathcal O(JT^{K+1})+\mathcal O(J^2T^2),
\end{align}
which verifies Eq.~\eqref{target}.

\subsection*{Efficient approximate construction}
\label{sec:approx-construction}

Let $f_1,\ldots,f_K:[0,1]\to[0,\infty)$ be continuous densities, and define $\mu_j(A):=\int_A f_j(\tau)d\tau$ and $F_j:=\mu_j([0,1])$.
An $\varepsilon$-consensus splitting among $q$ bins $A_1,\ldots,A_q$ satisfies $|\mu_j(A_c)-F_j/q|\le \varepsilon F_j/q$ for every $j=1,\ldots,K$ and $c=1,\ldots,q$.

Finding a splitting with nearly the minimum number of cuts is computationally difficult. Already for $q=2$ and some fixed $\varepsilon_0>0$, deciding whether an $\varepsilon_0$-consensus splitting exists with $K-1$ cuts is NP-hard~\cite{filosratsikas2018hardness}. Allowing a logarithmic overhead, the deterministic algorithm of Ref.~\cite{alon2021efficient}, given oracle access to the measures $\mu_j$, computes an $\varepsilon$-consensus splitting using
\begin{align}
L\le (q-1)K \left\lceil 2+\log_2\frac{3q}{\varepsilon} \right\rceil
\label{eq:approx-cut-count}
\end{align}
cuts in
$\mathcal O\left(\poly(q,K,\log(1/\varepsilon))\right)$ time.

For the DD construction, set $q=|\mathcal G|$ and $f_{m+1}(\tau)=\tau^m$ for $m=0,\ldots,K-1$, so that $F_{m+1}=1/(m+1)$. Label the bins by $g_c\in\mathcal G$ and define $\delta_{c,m}:=\int_{A_c}\tau^m d\tau-1/[q(m+1)]$. The consensus condition gives $|\delta_{c,m}|\le\varepsilon/[q(m+1)]$. Using $\sum_c\gamma_\alpha(g_c)=0$, the generalized moment becomes $M_{\alpha,m}=\sum_c\gamma_\alpha(g_c)\delta_{c,m}$, and therefore
\begin{align}
|M_{\alpha,m}| \le \sum_{c=1}^{q}|\delta_{c,m}| \le \frac{\varepsilon}{m+1}.
\label{eq:approx-moment-residual}
\end{align}
Each cut corresponds to one pulse. Hence, Eq.~\eqref{eq:approx-cut-count} with $q=|\mathcal G|$ gives the claimed constructive pulse bound. In the present setting, the required measure queries can be evaluated explicitly using $\int_a^b\tau^m d\tau=(b^{m+1}-a^{m+1})/(m+1)$.

\begin{figure}[t!]
  \centering
  \begin{minipage}[c]{0.48\linewidth}
    \centering
    \includegraphics[width=\linewidth]{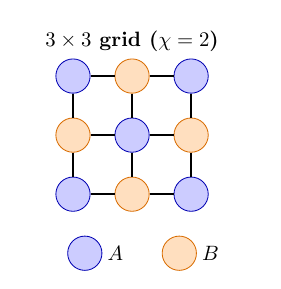}
  \end{minipage}
  \hfill
  \begin{minipage}[c]{0.50\linewidth}
    \centering
    \includegraphics[width=\linewidth]{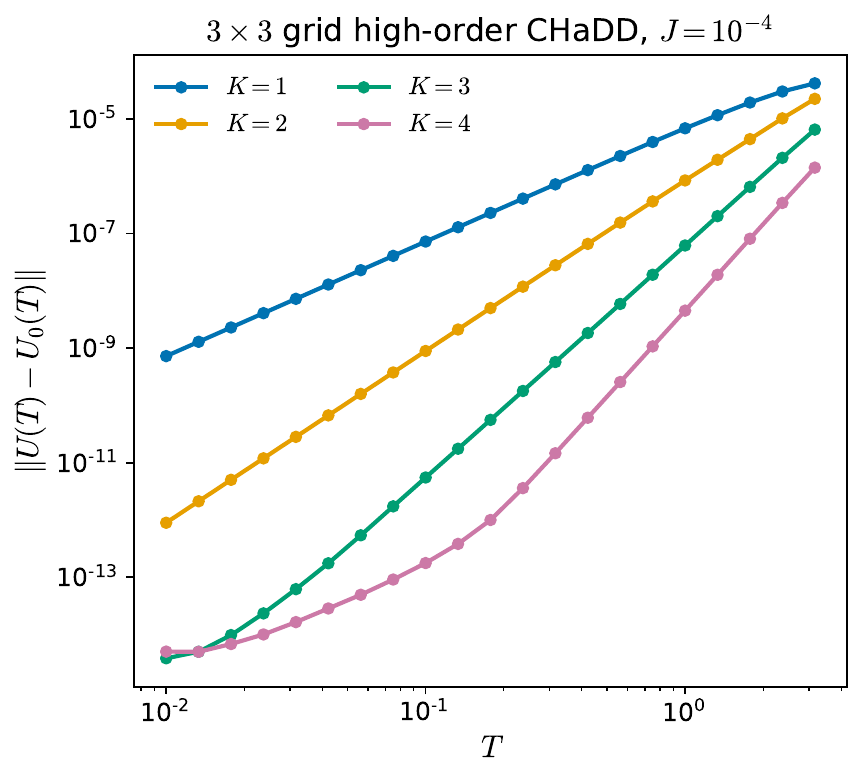}
  \end{minipage}
  \caption{{High-order CHaDD on a $3\times3$ open grid with $\chi=2$. Left: a $9$-qubit square lattice with graph-local $Z$-type edge noise. Standard CHaDD achieves first-order decoupling by alternating the collective $X$-pulses $X_A$ and $X_B$ on the two sublattices $A$ and $B$, respectively. Right: our construction yields a $3K$-pulse sequence achieving $K$th-order suppression.}}
  \label{fig:chadd_3x3}
\end{figure}

\subsection*{Higher-order CHaDD}
W consider a $9$-qubit system on a $3\times3$ open grid with single-axis couplings $H_{SB}=\sum_{v\in V} Z_v\otimes B_v+\sum_{(u,v)\in E} Z_uZ_v\otimes B_{uv}$ where $(u,v)$ runs over nearest-neighbor pairs on the grid. For this bipartite graph, CHaDD~\cite{brown2025efficient} achieves first-order decoupling by alternating collective $X$ pulses on sublattices $A$ and $B$. Applying Theorem~\ref{thm:existence} yields $L\le(2\chi-1)K=3K$ for $\chi=2$. The optimized sequence from the single-qubit case can be reused directly under $X\mapsto X_A$ and $Z\mapsto X_B$, giving a $3K(+1)$-pulse sequence that suppresses the leading error to order $K$, consistent with the numerical results in Fig.~\ref{fig:chadd_3x3}. More generally, arbitrary $2$-local noise on a graph of chromatic number $\chi$ admits high-order CHaDD sequences with $L=\mathcal O(\chi K)$ (see the SM for more details).

%%%%%%%%%%%%%%%%%%%%%%%%%%%%%%%%%%%%%%%%%%%%%%%%%%%%%%%%%%%%%%%%%%%%%%%%%%%%%%%
% -------------  SUPPLEMENTAL MATERIAL  -------------
%%%%%%%%%%%%%%%%%%%%%%%%%%%%%%%%%%%%%%%%%%%%%%%%%%%%%%%%%%%%%%%%%%%%%%%%%%%%%%%

\clearpage
\onecolumngrid
\pagenumbering{arabic}

\begin{center}
\textbf{\large Supplemental Material}
\end{center}

\setcounter{table}{0}
\renewcommand{\thetable}{S\arabic{table}}
\setcounter{figure}{0}
\renewcommand{\thefigure}{S\arabic{figure}}
\setcounter{equation}{0}
\renewcommand{\theequation}{S\arabic{equation}}
\setcounter{secnumdepth}{3}
\setcounter{theorem}{0}
\renewcommand{\thetheorem}{S\arabic{theorem}}

\section{Proof of Theorem \ref{thm:existence}}

\begin{theorem}[Existence of an order-$K$ moment-cancelling sequence with at most $(|\mathcal G|-1)K$ pulses; Restatement of Theorem~\ref{thm:existence}]\label{thm:existence_apx}
Let $H = H_0 + H_{SB}$ with $H_{SB} = \sum_\alpha \sigma_\alpha \otimes B_\alpha$ and let $\mathcal{G}$ be a decoupling group for $H$ {as defined in Eq.~\eqref{decoupling-group-condition}}. For any integer $K \ge 1$, there exist piecewise-constant functions $y_\alpha : [0,1] \to \{\pm 1\}$, indexed by all Pauli strings $\sigma_\alpha$ appearing in $H_{SB}$, such that:
\begin{enumerate}
\item\label{cond:i}
there exist (normalized) pulse timings $0 < \tau_1 < \cdots < \tau_L < 1$ with
\begin{align}\label{condition1}
L \le (|\mathcal{G}|-1)K 
\end{align}
for which all $y_\alpha(\tau)$ are constant on each subinterval $[0,\tau_1), [\tau_1,\tau_2), \ldots, [\tau_L,1]$, and 
\item\label{cond:ii}the generalized moments in \eqref{MomentsDef} vanish,
\begin{align}
M_{\alpha, m} = \int_0^1 y_\alpha(\tau) \tau^m d\tau = 0,
\end{align}
for all $\alpha$ and all $m=0,1,\dots,K-1$.
\end{enumerate}
\end{theorem}
%Equivalently, $[0,1]$ can be partitioned into at most $(|\mathcal{G}|-1)K+1$ segments on which all $y_\alpha(\tau)$ are constant, so such sign functions can be realized by a pulse sequence using at most $(|\mathcal{G}|-1)K$ Pauli pulses. 

Our construction hinges on the continuous necklace-splitting theorem:

\begin{lemma}[Necklace-splitting for continuous densities, Theorem 1.2 in \cite{alon1987splitting}]\label{lem:necklace}
Let $f_1,\ldots,f_K: [0,1]\to[0,\infty)$ be continuous functions and fix an integer $q\ge2$. There exist cut points $0<\tau_1<\cdots<\tau_L<1$ with $L\le (q-1)K$ and a partition of the subintervals $[0,\tau_1), [\tau_1,\tau_2), \dots, [\tau_L,1]$ into $q$ bins $A_1,\ldots,A_q$ such that, for every $j=1,\ldots,K$ and $c=1,\ldots,q$,
\begin{align}
\int_{A_c} f_j(\tau) d\tau=\frac{1}{q}\int_0^1 f_j(\tau) d\tau.
\end{align}
\end{lemma}
We note that the original theorem is stated for finite measures. Restricting to absolutely continuous measures $d\mu_j(\tau)=f_j(\tau)d\tau$ immediately gives the above statement.

%We now instantiate Lemma~\ref{lem:necklace} with $q = |\mathcal G|$ and $f_m(\tau)=\tau^m$ ($m=0,\ldots,K-1$), mapping the $q$ bins to elements of the decoupling group $\mathcal G$ to prove Theorem~\ref{thm:existence}.

To prove Theorem~\ref{thm:existence_apx}, apply Lemma~\ref{lem:necklace} with $q=|\mathcal G|$ and $f_m(\tau)=\tau^m$ for $m=0,\ldots,K-1$, and map the $q$ bins to elements of $\mathcal G$.

\begin{proof}[Proof of Theorem~\ref{thm:existence_apx}]
For each Pauli string $\sigma_\alpha$ in $H_{SB}$ and each $g\in\mathcal G$, conjugation yields
\begin{align}
g_\ell^\dagger \sigma_\alpha g_\ell = \gamma_\alpha(g_\ell) \sigma_\alpha, \qquad \gamma_\alpha(g_\ell)\in\{\pm1\}.
\label{chi-def}
\end{align}
By expanding $H_{SB} = \sum_\alpha \sigma_\alpha\otimes B_\alpha$ and using the decoupling group condition \eqref{decoupling-group-condition},
\begin{align}
0 &= \frac{1}{|\mathcal G|} \sum_{g\in\mathcal G} g^\dagger H_{SB} g = \frac{1}{|\mathcal G|} \sum_{\ell=1}^{|\mathcal G|} \sum_\alpha g_\ell^\dagger \sigma_\alpha g_\ell \otimes B_\alpha \\
&= \sum_\alpha \left(\frac{1}{|\mathcal G|}\sum_{\ell=1}^{|\mathcal G|} \gamma_\alpha(g_\ell)\right) \sigma_\alpha \otimes B_\alpha,
\end{align}
we see that
\begin{align}
\sum_{\ell=1}^{|\mathcal G|} \gamma_\alpha(g_\ell) = 0.
\label{sign-sum-zero-alpha}
\end{align}
Next, apply Lemma~\ref{lem:necklace} with $q = |\mathcal G|$ and the $K$ nonnegative continuous functions $f_m(\tau) = \tau^m$ ($m = 0,1,\ldots,K-1$). The lemma guarantees the existence of cut points $0 < \tau_1 < \cdots < \tau_L < 1$, with $L \le (|\mathcal G|-1)K$ and a partition of the induced subintervals into bins $A_1,\ldots,A_{|\mathcal G|}$ such that for all $m<K$ and all $c$,
\begin{align}
\int_{A_c} \tau^m d\tau = \frac{1}{|G|} \int_0^1 \tau^m d\tau.
\label{necklace-equal-mass}
\end{align}
Define $y_\alpha(\tau)$ by assigning to each bin $A_c$ the sign pattern determined by $g_c$
\begin{align}
y_\alpha(\tau) := \gamma_\alpha(g_c) \qquad \text{for } \tau \in A_c.
\end{align}
By construction, each $y_\alpha$ is piecewise constant and is constant on every subinterval of the partition $[0,\tau_1),[\tau_1,\tau_2),\dots,[\tau_L,1]$, {which satisfies Condition \eqref{cond:i}.}

It remains to verify the generalized moment constraints in {Condition \eqref{cond:ii}}. For every $\alpha$ and $m<K$,
\begin{align}
M_{\alpha,m}
&= \int_0^1 y_\alpha(\tau) \tau^m d\tau \\
&= \sum_{c=1}^{|\mathcal G|} \gamma_\alpha(g_c) \int_{A_c} \tau^m d\tau \label{secondline}\\
&= \frac{1}{|\mathcal G|} \left( \sum_{c=1}^{|\mathcal G|} \gamma_\alpha(g_c) \right) \int_0^1 \tau^m d\tau  \label{thirdline},\\
&= 0, \label{fourthline}
\end{align}
%where we used property \eqref{necklace-equal-mass} of the bins {to get Eq.~\eqref{thirdline}} and the decoupling group condition \eqref{sign-sum-zero-alpha} {to get Eq.~\eqref{fourthline}}.
where we used Eq.~\eqref{necklace-equal-mass} and Eq.~\eqref{sign-sum-zero-alpha}.
By Lemma~\ref{lem:necklace}, the number of cut points (and hence the number of pulses) is at most $L \le (|\mathcal G|-1)K$, and the number of constant segments is at most $L+1 \le (|\mathcal G|-1)K+1$. This completes the proof.
\end{proof}

\section{Proof of Lemma~\ref{lem:lowerbound}}

%In this section, we show that any DD sequence that cancels all generalized time-moments up to order $K-1$ must use at least $K$ pulses. Consequently, the $3K$-pulse sequence we construct for single-qubit general decoherence is asymptotically tight in $K$.

%\red{MM: copy the statement of lemma}

\begin{lemma}[$\Omega(K)$ lower bound; Restatement of Lemma~\ref{lem:lowerbound}]\label{lem:lowerbound_apx}
Let $y:[0,1]\to\{\pm1\}$ be piecewise constant with $r$ sign changes on $(0,1)$. If 
\begin{align}\label{lemma_condition_apdx}
\int_0^1 y(\tau) \tau^m d\tau=0\qquad (m=0,1,\ldots,K-1),
\end{align}
then necessarily $r\ge K$.
\end{lemma}

{The idea of the proof is as follows. Since Eq.~\eqref{lemma_condition_apdx} means that $y$ is orthogonal to every polynomial with degree $r$,  it suffices to find a single polynomial $q$ with $\deg q\le r$ such that $\int_0^1 y(\tau)q(\tau)d\tau\neq 0$. By choosing a polynomial $P$ that vanishes at all sign-flip points and setting $q:=P'$, the integral $\int_0^1 y(\tau)P'(\tau)d\tau$ telescopes to an endpoint term $P(1)$, which we can set to be nonzero. This yields the desired contradiction.}

\begin{proof}
Let the sign-flipping points be $0<\tau_1<\cdots<\tau_r<1$, and set $\tau_0:=0$, $\tau_{r+1}:=1$. Let $s_k\in\{\pm1\}$ be the constant value of $y$ on $(\tau_k,\tau_{k+1})$ for $k=0,\dots,r$. Define 
\begin{align}
P(\tau)=a \tau\prod_{j=1}^{r}(\tau-\tau_j),\qquad 
a:=\frac{s_r}{\prod_{j=1}^{r}(1-\tau_j)}.
\end{align}
Then $P(\tau_j)=0$ for $j=1,\dots,r$, $P(0)=0$, and $P(1)=s_r$. We choose $q(\tau)=P'(\tau)$; note that $\deg P=r+1$ and $\deg q\le r$.

We compute
\begin{align}
\int_{0}^{1} y(\tau) P'(\tau)  d\tau
& = \sum_{k=0}^{r} s_k \int_{\tau_k}^{\tau_{k+1}} P'(\tau)  d\tau \\
&= \sum_{k=0}^{r} s_k\left(P(\tau_{k+1})-P(\tau_k)\right).
\end{align}
Rearranging the sum gives
\begin{align}
& \sum_{k=0}^{r} s_k\left(P(\tau_{k+1})-P(\tau_k)\right) \\
& \quad = s_r P(1) - s_0 P(0) - \sum_{j=1}^{r}\left(s_j-s_{j-1}\right) P(\tau_j) \\
& \quad = 1,
\end{align}
where the last equality holds as (by construction) $P(\tau_j)=0$ for all $j$ and $P(0)=0$, while $P(1)=s_r$.

If $r\le K-1$ then $\deg P'\le K-1$, so we can write $P'(\tau)=\sum_{m=0}^{K-1}a_m\tau^m$ and thus
\begin{align}
\int_0^1 y(\tau) P'(\tau)  d\tau
= \sum_{m=0}^{K-1} a_m \int_0^1 y(\tau) \tau^m  d\tau
=0,
\end{align}
which contradicts the previous equality. Therefore $r\ge K$.
\end{proof}

\section{Extensions beyond the noise model used in the main text}

\subsection{Extension to time-dependent (classical) noise}
In this section, we briefly note that the same moment-cancellation mechanism extends to time-dependent (classical) Hamiltonian noise models that are widely used to describe environmental fluctuations, and for which DD performance has been extensively studied (see, e.g., Refs.~\cite{green2013arbitrary,bylander2011noise,alvarez2010performance,souza2011robust,ramon2022qubit}).
For concreteness, consider a single qubit with
\begin{align}
H(t)=H_c(t)+H_n(t),\qquad 
H_n(t)=\sum_{\alpha\in\{x,y,z\}}\beta_\alpha(t)\sigma_\alpha,
\label{classical_model}
\end{align}
where $\beta_\alpha(t)\in\mathbb{R}$ are unknown time-dependent coefficients and $H_c(t)$ implements ideal instantaneous Pauli pulses. Let $U_c(t)$ be the ideal control propagator generated by $H_c(t)$, and write the toggling-frame evolution as $\tilde U(T)=\mathcal{T}\exp\left[-i\int_0^T \tilde H_n(t)dt\right]$ with $\tilde H_n(t):=U_c^\dagger(t)H_n(t)U_c(t)$. For Pauli pulses, conjugation maps each Pauli to itself up to a sign on every free-evolution segment, so
\begin{align}
\tilde H_n(t)=\sum_{\alpha}y_\alpha(t)\beta_\alpha(t)\sigma_\alpha,\qquad y_\alpha(t)\in\{\pm1\}.
\label{classical_toggling}
\end{align}
Expanding $\tilde U(T)$ via the Magnus series, the leading term is
\begin{align}
\Omega_1(T)=-i\int_0^T \tilde H_n(t)dt
=-i\sum_\alpha \sigma_\alpha \int_0^T y_\alpha(t)\beta_\alpha(t)dt.
\label{classical_Omega1}
\end{align}
We can Taylor expand $\beta_\alpha(T\tau)$ in normalized time $\tau=t/T$. Assuming $\beta_\alpha$ is $K$-times differentiable on $[0,T]$, Taylor's theorem gives
\begin{align}
\beta_\alpha(T\tau)=\sum_{m=0}^{K-1}\frac{\beta_\alpha^{(m)}(0)}{m!}(T\tau)^m+\frac{(T\tau)^K}{K!}\beta_\alpha^{(K)}(\xi_{\alpha,\tau}),
\qquad \xi_{\alpha,\tau}\in[0,T\tau].
\label{classical_Taylor}
\end{align}
Substituting \eqref{classical_Taylor} into \eqref{classical_Omega1} shows that the coefficients of $T^{m+1}$ are proportional to the same generalized moments as in the main text,
\begin{align}
M_{\alpha,m}=\int_0^1 y_\alpha(\tau)\tau^m d\tau,
\label{classical_moments}
\end{align}
so imposing $M_{\alpha,m}=0$ for all $\alpha$ and $m=0,\dots,K-1$ cancels all contributions to $\Omega_1(T)$ up to order $K$. 

Define the noise-amplitude scale
\begin{align}
\beta_{\max}:=\sup_{t\in[0,T]}\sum_{\alpha\in\{x,y,z\}}|\beta_\alpha(t)|,
\label{classical_betamax}
\end{align}
so that $\|\tilde H_n(t)\|\le \beta_{\max}$ for all $t$. In addition, in the slow (low-bandwidth) regime we assume that $\beta_\alpha$ has an effective cutoff frequency $\omega_c$, in the sense that
\begin{align}
\sup_{t\in[0,T]}|\beta_\alpha^{(K)}(t)| =\mathcal O\big(\beta_{\max}\omega_c^{K}\big).
\label{classical_slow_assump}
\end{align}
For example, consider a single-frequency noise $\beta_\alpha(t)=A_\alpha\cos(\omega t+\phi_\alpha)$. Then $\sup_{t\in[0,T]}|\beta_\alpha(t)|=|A_\alpha|$ and $\sup_{t\in[0,T]}|\beta_\alpha^{(K)}(t)|=|A_\alpha||\omega|^K$. If $|\omega|\le \omega_c$ and $|A_\alpha|\le \beta_{\max}$, it follows that $\sup_{t\in[0,T]}|\beta_\alpha^{(K)}(t)|=\mathcal O(\beta_{\max}\omega_c^K)$ as in \eqref{classical_slow_assump}.

Then the Taylor remainder term in \eqref{classical_Taylor} implies
\begin{align}
\|\Omega_1(T)\| =\mathcal O\big(\beta_{\max}T(\omega_cT)^{K}\big).
\label{classical_Omega1_bound}
\end{align}
The second-order Magnus expansion is bounded by
\begin{align}
\|\Omega_2(T)\|=\mathcal O(\beta_{\max}^2T^2).
\label{classical_Omega2_bound}
\end{align}
Hence, in the small $\beta_{\max},T$ regime where the Magnus expansion converges,
\begin{align}
\|\tilde U(T)-I\| =\mathcal O\big(\beta_{\max}T(\omega_cT)^{K}\big)+\mathcal O(\beta_{\max}^2T^2),
\label{classical_final}
\end{align}
which matches the structure of the main text estimate in Eq.~\eqref{target}, with the system-bath coupling scale replaced by the classical noise amplitude scale.

{
\subsection{Generalized filter-function interpretation}

Although our analysis in the main text is formulated directly at the operator level, the same moment constraints admit a natural frequency-domain interpretation within the standard generalized filter-function and spectral-overlap formalisms \cite{green2013arbitrary,paz2014general,clausen2012task,gordon2007universal,kofman2004unified}.

For simplicity, we suppress the free system Hamiltonian in this subsection and write the toggling-frame interaction as
\begin{align}
\widetilde H_{SB}(t) = \sum_{\alpha\in\{x,y,z\}} y_\alpha(t)\sigma_\alpha \otimes B_\alpha(t),
\qquad
B_\alpha(t):=e^{iH_B t}B_\alpha e^{-iH_B t},
\label{ff_toggling_general}
\end{align}
where $y_\alpha(t)\in\{\pm1\}$ are the switching functions generated by the pulse sequence. Assume that the bath is initially in a stationary state $\rho_B$ and that the first moments vanish,
\begin{align}
\Tr\big(\rho_B B_\alpha(t)\big)=0.
\end{align}
Then the leading nontrivial contribution to the reduced dynamics is governed by the bath two-point correlation functions
\begin{align}
C_{\alpha\beta}(\tau)
:= \Tr\big(\rho_B B_\alpha(\tau)B_\beta(0)\big),
\label{ff_corr_def}
\end{align}
and, on the control side, by bilinear functionals of the switching functions \cite{clausen2012task,gordon2007universal,kofman2004unified}.

For a stationary bath, define the spectral-density matrix
\begin{align}
G_{\alpha\beta}(\omega) := \int_{-\infty}^{\infty} d\tau  e^{i\omega\tau} C_{\alpha\beta}(\tau),
\label{ff_spec_def}
\end{align}
and the Fourier amplitudes of the switching functions
\begin{align}
Y_\alpha(\omega,T) := \int_0^T dt y_\alpha(t)e^{i\omega t} = T\int_0^1 d\tau y_\alpha(\tau)e^{i\omega T\tau}.
\label{ff_amp_def}
\end{align}
In the standard generalized filter-function picture, the second-order coefficients of the reduced dynamics take the form
\begin{align}
\Gamma_{\alpha\beta}(T) \propto
\int_{-\infty}^{\infty}\frac{d\omega}{2\pi}
G_{\alpha\beta}(\omega) F_{\alpha\beta}(\omega,T),
\qquad
F_{\alpha\beta}(\omega,T) := Y_\alpha(\omega,T)Y_\beta^*(\omega,T),
\label{ff_overlap}
\end{align}
up to the system-operator tensor structure \cite{clausen2012task,gordon2007universal,kofman2004unified}. Thus $Y_\alpha(\omega,T)$ and $F_{\alpha\beta}(\omega,T)$ play the role of generalized filter amplitudes and quadratic filter functions associated with the pulse sequence \cite{green2013arbitrary,paz2014general}.

\subsubsection{Low-frequency zeros from moment cancellation}

In the low-frequency regime, the connection with our moment constraints is immediate. Using Eq.~\eqref{ff_amp_def} and expanding the exponential gives
\begin{align}
Y_\alpha(\omega,T) = T\sum_{m=0}^{\infty}\frac{(i\omega T)^m}{m!}
\int_0^1 d\tau y_\alpha(\tau)\tau^m = T\sum_{m=0}^{\infty}\frac{(i\omega T)^m}{m!}M_{\alpha,m}.
\label{ff_moment_expansion}
\end{align}
Therefore, if the moment constraints in Eq.~\eqref{moment-constraints-compact} hold, namely
\begin{align}
M_{\alpha,m}=0, \qquad m=0,1,\dots,K-1,
\end{align}
so, in the low-frequency regime $\omega T\ll 1$,
\begin{align}
Y_\alpha(\omega,T) = \mathcal O\big(T(\omega T)^K\big),
\qquad
F_{\alpha\beta}(\omega,T) = \mathcal O\big(T^2(\omega T)^{2K}\big).
\label{ff_lowfreq_zero}
\end{align}
In other words, the moment-cancellation conditions enforce a $K$th-order zero of each control Fourier amplitude at low frequency, and hence a $2K$th-order infrared zero of the associated quadratic filter function \cite{paz2014general}.

\subsubsection{Exact filter zeros at prescribed frequencies}

More generally, however, one need not restrict attention to low frequency. In the present single-qubit Pauli setting, the same necklace-splitting construction can be applied directly to Fourier test functions in order to impose exact zeros of the control Fourier amplitudes at any prescribed frequency, and more generally at any prescribed finite set of frequencies. To see this, fix a total evolution time $T$ and let
\begin{align}
\Omega=\{\omega_1,\dots,\omega_M\}\subset\mathbb R
\end{align}
be any finite target set. Consider the nonnegative continuous functions
\begin{align}
f_0(\tau)=1,\qquad
f^{(c)}_r(\tau)=1+\cos(\omega_rT\tau),\qquad
f^{(s)}_r(\tau)=1+\sin(\omega_rT\tau),
\qquad r=1,\dots,M.
\end{align}
Applying the continuous necklace-splitting theorem of Lemma \ref{lem:necklace} to this family of $2M+1$ functions with $q=4$ yields a partition of $[0,1]$ using at most
\begin{align}
L\le 3(2M+1)
\end{align}
cuts, together with an assignment of the resulting subintervals to the four Pauli sectors $g\in\{I,X,Y,Z\}$, such that every sector receives the same integral of each of these functions. Defining the switching pattern on the sector labeled by $g$ through
\begin{align}
y_\alpha(\tau)=\gamma_\alpha(g),
\qquad
g^\dagger \sigma_\alpha g = \gamma_\alpha(g)\sigma_\alpha,
\qquad
\gamma_\alpha(g)\in\{\pm1\},
\end{align}
and using the relation $\sum_{g\in\{I,X,Y,Z\}}\gamma_\alpha(g)=0$, one obtains
\begin{align}
\int_0^1 d\tau y_\alpha(\tau)=0,
\qquad
\int_0^1 d\tau y_\alpha(\tau)\cos(\omega_rT\tau)=0,
\qquad
\int_0^1 d\tau y_\alpha(\tau)\sin(\omega_rT\tau)=0
\end{align}
for every $\alpha\in\{x,y,z\}$ and every $r=1,\dots,M$. Hence
\begin{align}
Y_\alpha(\omega_r,T)=T\int_0^1 d\tau y_\alpha(\tau)e^{i\omega_rT\tau}=0,
\qquad
F_{\alpha\beta}(\omega_r,T)=0,
\end{align}
for all $\alpha,\beta\in\{x,y,z\}$ and all $r=1,\dots,M$. Thus, beyond the low-frequency Taylor-zero picture implied by the moment constraints, the same combinatorial construction can be used to synthesize exact filter-function notches at arbitrary prescribed frequencies, with no low-frequency assumption.

\subsubsection{Pure-dephasing spin-boson model}
As a concrete example where the filter-function interpretation becomes exact, consider the pure-dephasing spin-boson model
\begin{align}
H = H_B + Z\otimes B_z, \qquad H_B=\sum_k \omega_k b_k^\dagger b_k, \qquad B_z=\sum_k \bigl(g_k b_k+g_k^* b_k^\dagger\bigr).
\label{sb_pd_model}
\end{align}
Under ideal instantaneous Pauli $X$ pulses, the toggling-frame interaction becomes
\begin{align}
\widetilde H_{SB}(t) = y_z(t) Z \otimes B_z(t),
\qquad B_z(t):=e^{iH_B t}B_z e^{-iH_B t}.
\end{align}
In this case the coherence of an evolved state is known exactly \cite{uhrig2007keeping,uhrig2008exact}. Writing the coherence as
\begin{align}
\rho_{01}(T)=\rho_{01}(0)e^{-\chi(T)},
\end{align}
the decoherence exponent can be expressed in spectral-overlap form as
\begin{align}
\chi(T) = \int_0^\infty d\omega
J_z(\omega)\coth\left(\frac{\beta\omega}{2}\right)
|Y_z(\omega,T)|^2,
\label{sb_pd_exact_chi}
\end{align}
where $J_z(\omega)$ denotes the corresponding bath spectral density and
\begin{align}
Y_z(\omega,T)=\int_0^T dt y_z(t)e^{i\omega t}
\end{align}
is precisely the $\alpha = z$ specialization of the control Fourier amplitude introduced in Eq.~\eqref{ff_amp_def}. Equivalently, in the more conventional notation one may write $F(\omega T)\propto \omega^2 |Y_z(\omega,T)|^2$ which is often referred to as filter function \cite{cywinski2008how,uhrig2008exact}.

Hence, if the moment constraints hold in the dephasing channel,
\begin{align}
\int_0^1 d\tau y_z(\tau)\tau^m=0, \qquad m=0,1,\dots,K-1,
\end{align}
then Eq.~\eqref{ff_lowfreq_zero} implies
\begin{align}
Y_z(\omega,T)=\mathcal O\big(T(\omega T)^K\big),
\end{align}
and therefore the exact decoherence exponent acquires a $2K$th-order infrared suppression. If, in addition, $\int_0^\infty d\omega J_z(\omega)\coth\left(\frac{\beta\omega}{2}\right)\omega^{2K}<\infty$, then
\begin{align}
\chi(T)=\mathcal O(T^{2K+2}), \qquad T\to 0.
\label{sb_pd_scaling}
\end{align}

A canonical example is the Uhrig dynamical decoupling (UDD) sequence \cite{uhrig2007keeping,uhrig2008exact}, with pulse locations $t_j=T\sin^2\left(\frac{j\pi}{2K+2}\right)$ for $j=1,\dots,K$ for which the first $K$ derivatives of the corresponding dephasing control amplitude vanish at $\omega=0$.

This suggests that the same low-frequency suppression mechanism may extend beyond pure dephasing to more general spin-boson noise models \cite{zhao2014spin,hahn2025efficiency}.

}

\section{Numerical Simulation Details and Additional Simulations}
\subsection{Numerical optimization procedure}
In the single-qubit example in the main text, the goal is to construct a Pauli pulse sequence that satisfies the moment constraints up to order $K-1$ while using only $3K$ $(+1)$ pulses. We fix a pulse pattern a priori and optimize only the inter-pulse intervals.

For a given order $K$, we fix the number of inter-pulse intervals to be $L = 3K+1$ and choose a fixed sequence of control pulses $X \to Z \to X \to Z \to \cdots$. (We note that different pulse orderings may lead to either shorter sequences or sequences of the same length with improved robustness to pulse errors. Exploring such optimizations is left for future work). This generates control frames $f_1,\dots,f_L \in G=\{I,X,Y,Z\}$. Each frame $f_\ell$ is associated with a ``sign triple''
\begin{align}
\mathbf{s}_\ell = (s_{\ell,x},s_{\ell,y},s_{\ell,z}) \in \{\pm 1\}^3,
\qquad
\ell=1,\dots,L,
\end{align}
which specifies the sign of the Pauli operators $\sigma_\alpha$ in the toggling frame during the $\ell$-th interval. Equivalently, this defines switching functions $y_\alpha(\tau)$ for $\alpha\in\{x,y,z\}$ that are piecewise constant on $[0,1]$: $y_\alpha(\tau)=s_{\ell,\alpha}$ for $\tau\in[t_{\ell-1},t_\ell)$, where $0 = t_0 < t_1 < \cdots < t_L = 1$ are the normalized segment boundaries.

Rather than performing any numerical quadrature, we evaluate the integrals analytically. Since $y_\alpha(\tau)$ is piecewise constant,
\begin{align}
\int_0^1  \tau^m y_\alpha(\tau)d\tau = \sum_{\ell=1}^L s_{\ell,\alpha}\int_{t_{\ell-1}}^{t_\ell}\tau^m d\tau = \frac{1}{m+1} \sum_{\ell=1}^L s_{\ell,\alpha} \bigl(t_\ell^{m+1}-t_{\ell-1}^{m+1}\bigr),
\end{align}
and the overall factor $1/(m+1)$ is irrelevant for the location of the zeros. We therefore collect all moment residuals into a single vector $r(\Delta\tau)\in\mathbb{R}^{3K}$ whose components are
\begin{align}
r_{\alpha,m}(\Delta\tau) = \sum_{\ell=1}^L s_{\ell,\alpha}
\bigl(t_\ell^{m+1}-t_{\ell-1}^{m+1}\bigr), \qquad \alpha\in\{x,y,z\}, \quad m=0,\dots,K-1,
\end{align}
and minimize the squared norm
\begin{align}
\Phi(\Delta\tau)
= \sum_{\alpha,m} \bigl|r_{\alpha,m}(\Delta\tau)\bigr|^2,
\label{moment-cost-appendix}
\end{align}
subject to $\Delta\tau_\ell>0$ and $\sum_\ell \Delta\tau_\ell=1$. %Up to an overall constant, the objective \eqref{moment-cost-appendix} is equivalent to the cost function defined in \eqref{moment-cost-main} of the main text.

To enforce the simplex constraint, we parameterize the intervals by an unconstrained vector $\theta\in\mathbb{R}^L$ via a softmax map,
\begin{align}
\Delta\tau_\ell(\theta) = \frac{e^{\theta_\ell}}{\sum_{j=1}^L e^{\theta_j}}, \qquad \ell=1,\dots,L,
\end{align}
and treat $r(\theta):=r(\Delta\tau(\theta))$ as a nonlinear map $\mathbb{R}^L\to\mathbb{R}^{3K}$. The optimization problem is then formulated as a nonlinear least-squares problem for $\theta$, which we solve with a standard (trust-region reflective) algorithm. Specifically, we use the \texttt{least\_squares} function from \textsc{SciPy}. We use stopping tolerances on the function value, step size, and gradient ($\texttt{ftol} = \texttt{xtol} = \texttt{gtol} = 10^{-15}$) and allow up to $\texttt{max\_nfev}=10^5 := N_{\mathrm {fev}}$ function evaluations per run. %We make our code publicly available in \cite{}. 

For a given order $K$, the nonlinear least-squares problem has $3K+1$ variables and $3K$ residuals. Each evaluation of the moment residual vector $r(\theta)$ costs $\mathcal O(K^2)$ operations, and the trust-region reflective solver uses at most $N_{\mathrm{fev}}$ function evaluations per run. The total classical computational cost therefore scales as $\mathcal O (N_{\mathrm{fev}} K^2)$. In general, if one chooses to apply the same numerical optimization to a DD sequence constructed from a decoupling group $\mathcal{G}$ with $\mathcal O(|\mathcal{G}|K)$ segments, the classical computational cost scales as $\mathcal{O} \bigl(N_{\mathrm{fev}} (|\mathcal{G}|K)^2\bigr)$.

\begin{figure}[t!]
    \centering
    \includegraphics[width=1\linewidth]{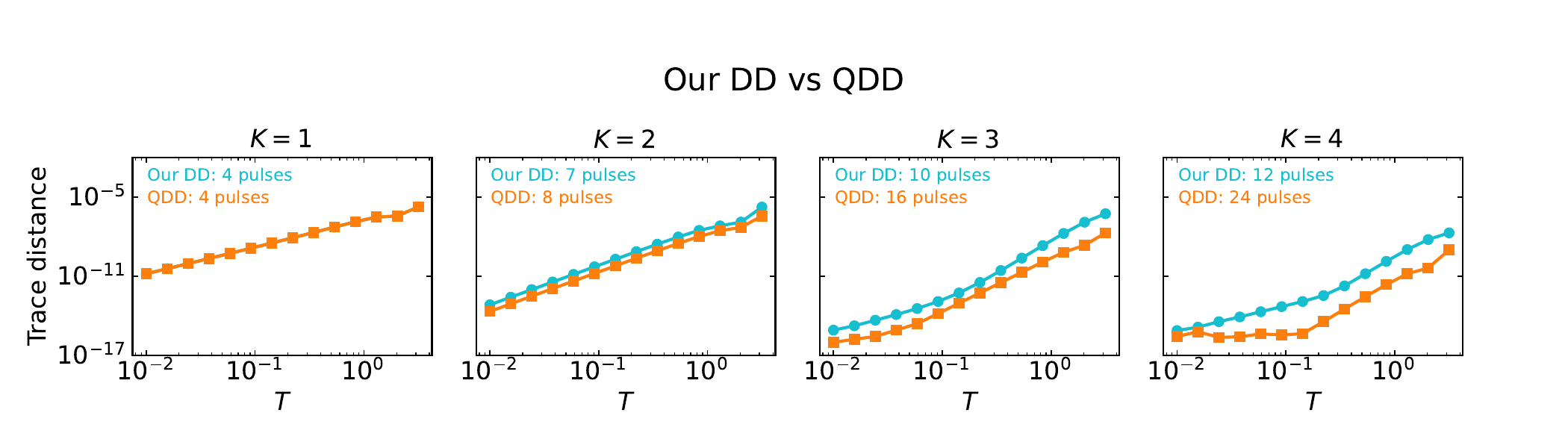}
        \caption{Comparison between our order-$K$ decoupling sequence and QDD in the same single-qubit model used in the main text with $\beta=1$ and $J=10^{-5}$. The plot shows the trace-distance error between the actual and ideal system states as a function of total evolution time $T$.}
    \label{fig:qdd_comparison}
\end{figure}

\begin{figure*}[t]
  \centering
  \includegraphics[width=1\textwidth]{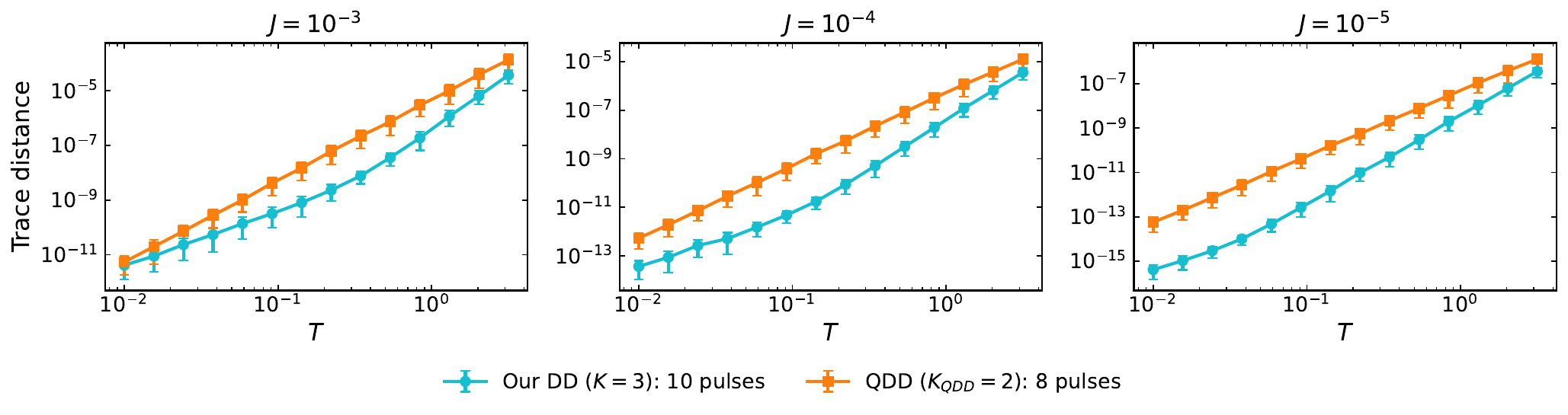}
  \caption{Comparison between our DD sequences and QDD for the single-qubit general decoherence model. We show the average trace distance error of the reduced system state versus $T$ for several small coupling strengths $J$, comparing our $K=3$ DD (10 pulses) with QDD at $K_{\rm QDD}=2$ (8 pulses). Each data point is averaged over 100 random initial product states.}
  \label{fig:QDD_comparison_SM}
\end{figure*}

\subsection{Comparison with QDD: additional simulations}\label{QDD_comparison}
In the main text, we benchmarked our construction against Quadratic Dynamical Decoupling (QDD) \cite{west2010near}. QDD is obtained by nesting two single-axis Uhrig sequences along orthogonal control axes (here taken to be $X$ and $Z$), and suppresses arbitrary single-qubit noise to order $K_{\rm QDD}$ using at most $(K_{\rm QDD}+1)^2$ pulses. QDD is widely considered to be (nearly) optimal for suppressing single-qubit noise.

In this section, we provide additional simulations on the comparison for various orders $K$. We use the same numerical setup as the numerical section in the main text, with the weak-coupling parameters fixed to $\beta = 1$ and $J = 10^{-5}$. We choose the initial state $\ket{+}$ for both the system and the bath, evolve under the noisy dynamics with a given DD sequence, and compute the trace distance between the reduced system state and the ideal state obtained in the absence of system-bath coupling. This error is plotted as a function of the total evolution time $T$. Figure \ref{fig:qdd_comparison} shows results for $K=K_{\rm QDD} =1,2,3,4$. Our DD sequence achieves the expected error scaling using only $3K(+1)$ pulses, compared to $(K_{\rm QDD}+1)^2$ pulses for QDD. 

{Finally, in addition to the comparisons presented in the main text (Fig.~\ref{fig:QDD_comparison_main}), we benchmark performance at a comparable resource budget by comparing our sequence ($K=3$, 10 pulses) with a QDD sequence of similar pulse count ($K_{\rm QDD} = 2$, 8 pulses) under the same setup. Although our protocol uses two additional pulses, our protocol achieves lower trace distance across all tested ranges of $T$ and $J$, as shown in Fig.~\ref{fig:QDD_comparison_SM}.}

\subsection{Extended data of optimized pulse timings and effect of imprecise pulse timing}

\begin{table}[h]
  \caption{Normalized segment lengths obtained by minimizing $\Phi(\Delta\tau)$.
  The pulse sequence is $X \to Z \to X \to Z \cdots$.
  Values are shown to 15 decimal places.}
  \label{optim_SM}
  {\renewcommand{\arraystretch}{1.2}
  \begin{ruledtabular}
  \begin{tabular}{@{}c l@{}}
    $K$ & $\{\Delta\tau_i\}_{i=1}^{L}$ \\
    \colrule

    % -------------------- K = 1 --------------------
    \parbox[t]{1.6em}{\centering 1} &
    \parbox[t]{0.72\linewidth}{\raggedright
      [0.250000000000000, 0.250000000000000, 0.250000000000000, 0.250000000000000]\par} \\

    % -------------------- K = 2 --------------------
    \parbox[t]{1.6em}{\centering 2} &
    \parbox[t]{0.72\linewidth}{\raggedright
      [0.078464834591372, 0.124999999999999, 0.171535165408626,
       0.250000000000001, 0.171535165408628, 0.125000000000001,
       0.078464834591374]\par} \\

    % -------------------- K = 3 --------------------
    \parbox[t]{1.6em}{\centering 3} &
    \parbox[t]{0.72\linewidth}{\raggedright
      [0.032292826653268, 0.063198676159276, 0.098322279339589,
       0.151677720660432, 0.154508497187470, 0.154508497187474,
       0.151677720660411, 0.098322279339567, 0.063198676159263,
       0.032292826653251]\par} \\

    % -------------------- K = 4 --------------------
    \parbox[t]{1.6em}{\centering 4} &
    \parbox[t]{0.72\linewidth}{\raggedright
      [0.015740558567996, 0.034840993778921, 0.058326727422055,
       0.092639337921122, 0.109259441431874, 0.122519668299637,
       0.133346545155644, 0.122519668299808, 0.109259441431969,
       0.092639337921442, 0.058326727422301, 0.034840993779070,
       0.015740558568160]\par} \\

    % -------------------- K = 5 --------------------
    \parbox[t]{1.6em}{\centering 5} &
    \parbox[t]{0.72\linewidth}{\raggedright
      [0.008599239338160, 0.020714620114750, 0.036488791784096,
       0.059197348763041, 0.075300583772990, 0.089748095447975,
       0.103048492653219, 0.106902828125843, 0.106902828125842,
       0.103048492653202, 0.089748095447950, 0.075300583772968,
       0.059197348763007, 0.036488791784073, 0.020714620114735,
       0.008599239338148]\par} \\

    % -------------------- K = 6 --------------------
    \parbox[t]{1.6em}{\centering 6} &
    \parbox[t]{0.72\linewidth}{\raggedright
      [0.005099840472689, 0.013081296294727, 0.023951104697432,
       0.039541779209467, 0.052702696284812, 0.065297618986759,
       0.077549425157617, 0.085458220790868, 0.090696933388149,
       0.093242169437664, 0.090696933387990, 0.085458220790698,
       0.077549425157258, 0.065297618986355, 0.052702696284446,
       0.039541779208967, 0.023951104697093, 0.013081296294494,
       0.005099840472515]\par} \\

    % -------------------- K = 7 --------------------
    \parbox[t]{1.6em}{\centering 7} &
    \parbox[t]{0.72\linewidth}{\raggedright
      [0.003218107155529, 0.008668617548996, 0.016370484427467,
       0.027437910209365, 0.037772854809786, 0.048100983157502,
       0.058499818469824, 0.066844362503103, 0.073589329135886,
       0.078650108193099, 0.080847424390996, 0.080847424390950,
       0.078650108192911, 0.073589329135586, 0.066844362502754,
       0.058499818469383, 0.048100983157064, 0.037772854809394,
       0.027437910208959, 0.016370484427200, 0.008668617548824,
       0.003218107155422]\par} \\

    % -------------------- K = 8 --------------------
    \parbox[t]{1.6em}{\centering 8} &
    \parbox[t]{0.72\linewidth}{\raggedright
      [0.002131868452877, 0.005973589234768, 0.011573393450091,
       0.019658348833883, 0.027726790026254, 0.036051120042435,
       0.044640380218587, 0.052270651341861, 0.059000709140531,
       0.064723118015122, 0.068786226300013, 0.071323172469493,
       0.072281264701016, 0.071323172476857, 0.068786226312206,
       0.064723118034568, 0.059000709164710, 0.052270651368313,
       0.044640380249883, 0.036051120072752, 0.027726790054342,
       0.019658348862506, 0.011573393469221, 0.005973589247443,
       0.002131868460270]\par} \\

  \end{tabular}
  \end{ruledtabular}}
\end{table}

%In Table \ref{optim_SM}, we provide more precise values of the pulse intervals reported in Table \ref{optim} of the main text, and show the results up to $K=8$.

In Table \ref{optim_SM}, we provide precise values of the pulse intervals up to $K=8$.

In practice, the pulse timings of the sequence shown in Table \ref{optim_SM} may not be implemented exactly, but can include some perturbations. Let $0=\tau_0<\tau_1<\cdots<\tau_L=1$ be the ideal (normalized) pulse times, and model the implemented times as
\begin{align}
\tilde{\tau}_k = \tau_k + \delta \tau_k, \qquad |\delta \tau_k| \le \varepsilon,
\end{align}
where $\varepsilon$ describes the timing precision.

In physical time $t\in[0,T]$ the ideal pulses are applied at $t_k=\tau_k T$, and the corresponding toggling-frame Hamiltonian $\widetilde H(t)$ is piecewise constant, changing only at the switching times. With timing errors the pulses are applied at $\tilde t_k=\tilde\tau_k T$, giving an implemented Hamiltonian $\widetilde H_\varepsilon(t)$. A straightforward calculation then yields
\begin{align}
\|\widetilde U_\varepsilon(T)-\widetilde U(T)\| \le \int_0^T \|\widetilde H_\varepsilon(t)-\widetilde H(t)\|dt = \mathcal O(J\varepsilon LT),
\end{align}
where $\widetilde U(T)$ and $\widetilde U_\varepsilon(T)$ denote the ideal and implemented evolution operators in the toggling frame. This argument only uses the piecewise-constant structure of the control, and therefore applies generically to DD sequences. 

To check this, we perform numerical simulations (with the same setups as in Section \ref{QDD_comparison}), but now truncate the decimal digits of the optimized pulse timings in Table \ref{optim_SM}, as well as for QDD. For fixed $J$, $K$, and a given number $d$ of kept decimal digits, we vary $T$. A simple estimate suggests that the crossover time $T_c$ at which the truncated sequence starts to deviate from the ideal (full-precision) error curve is determined by
\begin{align}
J T_c^{K+1} \sim J \varepsilon L T_c,
\end{align}
where $L$ is the number of pulses and $\varepsilon$ is the typical timing error induced by truncation. This gives
\begin{align}
T_c \sim (\varepsilon L)^{1/K},
\end{align}
so that if $\varepsilon$ scales as $10^{-d}$, we expect $T_c \propto 10^{-d/K}$. This scaling behavior is consistent with the numerical results shown in Figs.~\ref{fig:trunc_DD} and \ref{fig:trunc_QDD}, where the deviation point between the truncated and ideal curves moves systematically as the number of retained digits is increased.

\begin{figure}[t]
  \centering
  \includegraphics[width=0.95\textwidth]{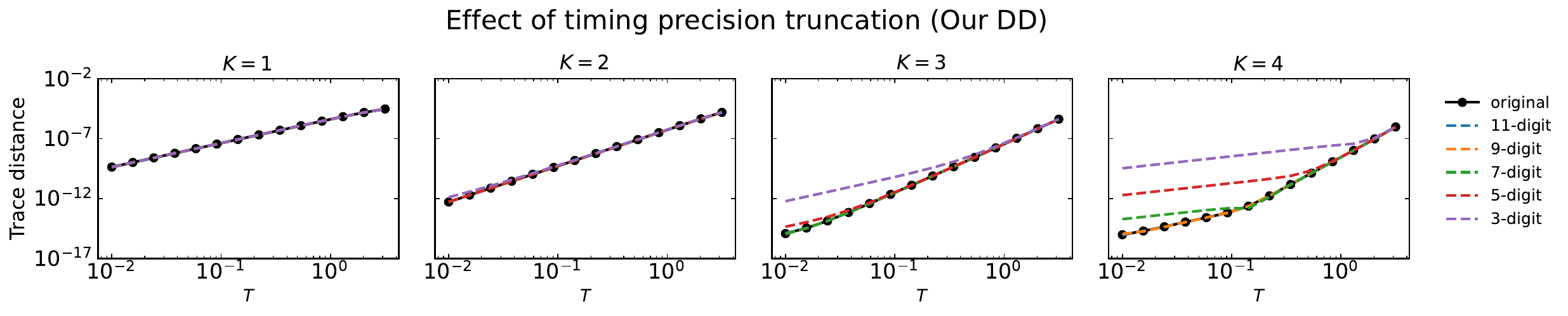}
  \caption{Error vs. $T$ for our DD sequence with various digit truncations.}
  \label{fig:trunc_DD}
\end{figure}

\begin{figure}[t]
  \centering
  \includegraphics[width=0.95\textwidth]{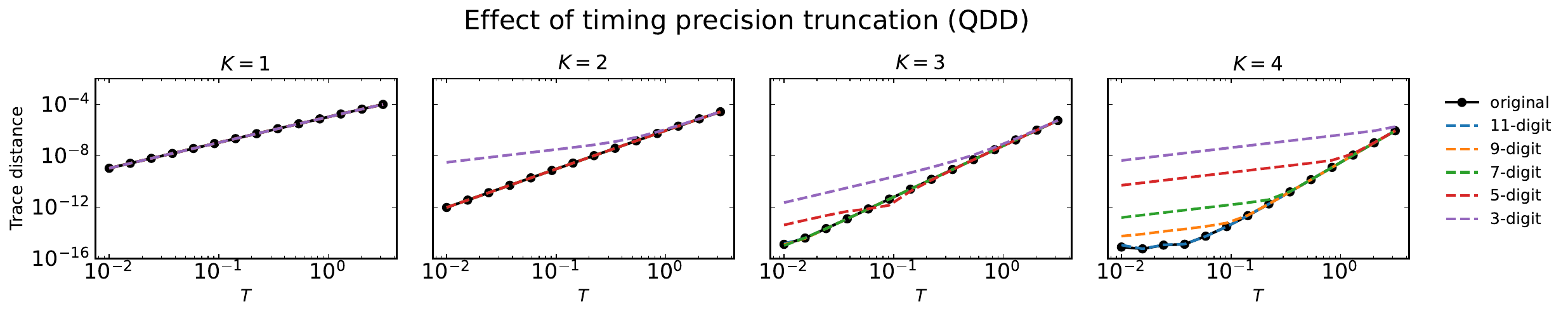}
  \caption{Same analysis as Fig.~\ref{fig:trunc_DD} for QDD, showing consistent $T_c$ behavior.}
  \label{fig:trunc_QDD}
\end{figure}

{
\subsection{Robust variants under coherent pulse errors and comparison with other robust DD}
%\lk{This newly added subsection addresses Referee A’s major comments regarding the effect of coherent errors and comparison to UR DD sequence.}

In this section, we study the effect of coherent pulse errors arising from, for simplicity, a uniform over- or under-rotation in every pulse. We first note that our sequence is already first-order robust against such coherent pulse errors. We then show that this robustness can be promoted to higher order simply by alternating the sign of the rotation angle, without increasing the total sequence length. As a representative example, we apply this construction to the 16-length DD sequence used in Fig.~\ref{fig:QDD_comparison_main} of the main text and compare its performance with quadratic DD and with the robust phase-shifted XY16 sequence of Ref.~\cite{souza2011robust}, which exactly cancels these pulse errors.

First, we briefly note that the alternating single-qubit pulse pattern used in our explicit constructions is automatically first-order robust against a uniform coherent over/under-rotation error. Let the ideal pulses be
\begin{align}
P_j=\exp\left(-i\frac{\pi}{2}\sigma_{a_j}\right), \qquad 
a_j= \begin{cases}
 x,& j\ \text{odd},\\
 z,& j\ \text{even},
 \end{cases}
\end{align}
so that, up to an irrelevant global phase, the ideal pulse word is
$X\to Z\to X\to Z\to\cdots$. The same argument applies, without modification, to any pair of orthogonal single-qubit Pauli axes.

Following the standard coherent pulse-error model used in analyses of robust DD sequences~\cite{maudsley1986modified,ryan2010robust,souza2011robust,souza2012robust,genov2017arbitrarily,ezzell2023dynamical}, we replace each ideal Pauli pulse by
\begin{align}
\widetilde P_j = \exp\left(-i(1+\varepsilon_r)\frac{\pi}{2}\sigma_{a_j}\right),
\qquad |\varepsilon_r|\ll 1,
\label{eq:coherent_pulse_error_model}
\end{align}
where the same relative rotation error $\varepsilon_r$ is applied to every pulse.

To analyze the resulting control error, let $F_1=I$ and define recursively
\begin{align}
F_{j+1}=P_jF_j,
\end{align}
so that $F_j$ is the ideal control frame immediately before the $j$th pulse. Writing
\begin{align}
\widetilde P_j = E_j P_j, \qquad E_j:=\exp\left(-i\varepsilon_r\frac{\pi}{2}\sigma_{a_j}\right),
\end{align}
the net control error after $L$ pulses is
\begin{align}
V_L &:= (P_L\cdots P_1)^\dagger (\widetilde P_L\cdots \widetilde P_1) \\ 
&= \exp\left(-i\varepsilon_r\frac{\pi}{2}F_L^\dagger \sigma_{a_L} F_L\right) \cdots \exp\left(-i\varepsilon_r\frac{\pi}{2}F_1^\dagger \sigma_{a_1} F_1\right).
\label{eq:control_error_expansion}
\end{align}
Expanding to first order in $\varepsilon_r$ gives
\begin{align}
V_L = I - i\varepsilon_r\frac{\pi}{2} \sum_{j=1}^{L} F_j^\dagger \sigma_{a_j} F_j + \mathcal O(\varepsilon_r^2).
\label{eq:control_error_first_order}
\end{align}

Consider one complete $XZXZ$ block. The corresponding toggled error generators are
\begin{align}
F_1^\dagger X F_1 &= X, \quad F_2^\dagger Z F_2 &= XZX = -Z, \quad F_3^\dagger X F_3 &= (ZX)^\dagger X (ZX) = -X, \quad F_4^\dagger Z F_4 &= (XZX)^\dagger Z (XZX) = Z,
\end{align}
and therefore
\begin{align}
\sum_{j=1}^{4} F_j^\dagger \sigma_{a_j} F_j
= X-Z-X+Z = 0.
\end{align}
Hence every complete four-pulse block carries no linear term in $\varepsilon_r$. Equivalently, any sequence formed by concatenating complete alternating $XZXZ$ blocks satisfies
\begin{align}
V_L = I + \mathcal O(\varepsilon_r^2).
\end{align}
This is the same first-order self-compensation mechanism that underlies the robustness of XY4-type cycles against these errors \cite{maudsley1986modified,ryan2010robust,souza2011robust,genov2017arbitrarily}.

\begin{figure}[t!]
  \centering
  \includegraphics[width=0.95\textwidth]{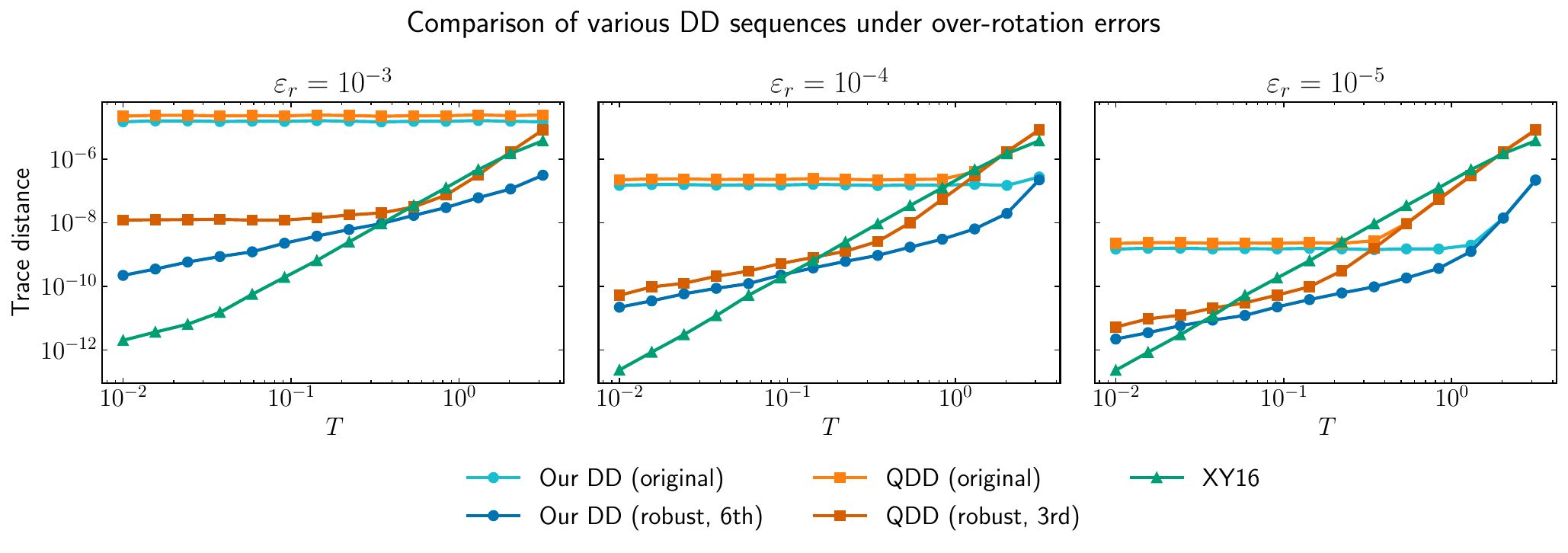}
  \caption{{Comparison between several length-16 DD sequences in the presence of uniform fixed over-rotations with physically motivated values of $\varepsilon_r$, at fixed $J=10^{-4}$. We compare the robust variants of our DD sequence with the robust XY16 sequence \cite{souza2011robust}, which fully cancels the error due to $\varepsilon_r$, and with a robust variant of QDD. Whereas the original sequence exhibits an error floor set by $\varepsilon_r$, the robust variant performs significantly better and yields the smallest error among all sequences across a broad range of parameters.}}
  \label{fig:robust_DD}
\end{figure}

\subsubsection{High-order robust variants for $K=5$ sequence and numerical simulation}

This first-order robustness can be promoted to higher order by allowing signed $\pi/2$ pulses,
\begin{align}
P_j^{(s_j)}:=\exp\left(-i s_j \frac{\pi}{2}\sigma_{a_j}\right),\qquad
\widetilde P_j^{(s_j)}:=\exp\left(-i s_j (1+\varepsilon_r)\frac{\pi}{2}\sigma_{a_j}\right),
\qquad s_j\in\{\pm1\},
\end{align}
where $\bar X$ and $\bar Z$ denote $s_j=-1$. Since $P_j^{(-1)}=-P_j^{(+1)}$, changing the sign does not alter the ideal DD action, except for an overall phase. However, it does modify the coherent pulse-error propagator. With $\eta:=\frac{\pi}{2}\varepsilon_r$, Eq.~\eqref{eq:control_error_expansion} becomes
\begin{align}
V_L(\mathbf s)=
\exp\left(-i\eta s_L F_L^\dagger \sigma_{a_L} F_L\right)\cdots
\exp\left(-i\eta s_1 F_1^\dagger \sigma_{a_1} F_1\right),
\end{align}
so an appropriate choice of the sign pattern $\mathbf s=(s_1,\dots,s_L)$ can cancel higher-order terms in $\varepsilon_r$ without increasing the sequence length.

For the 16-pulse sequence used in Fig.~\ref{fig:QDD_comparison_main} (i.e., $K=5$ in Table \ref{optim_SM}), one convenient sixth-order-robust choice is
\begin{align}
\mathcal S_6= X \bar Z X \bar Z  \bar X \bar Z \bar X \bar Z 
\bar X Z \bar X Z  X Z X Z .
\label{eq:S6_sequence}
\end{align}
A direct expansion gives
\begin{align}
V_{16}^{(\mathcal S_6)} = I-i 16\eta^6\sigma_y+\mathcal O(\eta^7)
= I-i \frac{\pi^6}{4}\varepsilon_r^6\sigma_y+\mathcal O(\varepsilon_r^7).
\label{eq:S6_error_scaling}
\end{align}
Thus, all coherent pulse-error terms up to order $\varepsilon_r^5$ vanish, while the pulse count and the DD performance in the ideal pulse limit remains unchanged. 

To benchmark the sign-flipped sequence \eqref{eq:S6_sequence}, we compare it with QDD and with a standard robust DD protocol at the same pulse budget. As the robust-DD reference, we use the (phase-shifted) XY16 sequence~\cite{souza2011robust},
\begin{align}
X Y X Y  Y X Y X  \bar X \bar Y \bar X \bar Y  \bar Y \bar X \bar Y \bar X ,
\end{align}
where $\bar Y$ denotes the negative-angle $Y$ pulse. This sequence uses the same pulse budget of $16$ pulses and completely cancels any $\varepsilon_r$ error. For QDD with $K_{\rm QDD}=3$, we similarly allow sign flips and exhaustively search over all $2^{16}$ choices. The best sign-flipped variant we find is
\begin{align}
X X \bar X Z  X X \bar X Z  \bar X X \bar X Z  \bar X \bar X X Z ,
\end{align}
which is third-order robust, so that its leading coherent pulse-error term scales as $\mathcal O(\varepsilon_r^4)$.

As a hardware-motivated scale for $\varepsilon_r$, we use the recent Google Sycamore characterization of coherent control errors by Ref.~\cite{gross2024characterizing}. That work reports sub-milliradian precision in estimating coherent gate parameters on a Sycamore processor, which places the residual coherent angle error naturally at the milliradian scale. In our convention, the physical rotation-angle error is $\delta \theta = \pi \varepsilon_r$; thus a representative $\delta \theta \sim 10^{-3}$ rad corresponds to $\varepsilon_r \sim 3 \times 10^{-4}$. We therefore take $\varepsilon_r \in [10^{-5}, 10^{-3}]$ for well-calibrated superconducting hardware. Fig.~\ref{fig:robust_DD} shows the resulting comparison of the trace distance error for the different length-16 DD sequences across this parameter range. While the original sequence exhibits the error floor seen in figure due to $\varepsilon_r$, the robust variant performs substantially better and, over a broad range of parameters, yields the smallest error among all sequences considered. This suggests that equally simple robust variants can be designed for other values of the order $K$, thereby improving the robustness of our DD family to such control errors and strengthening its practical relevance.

\begin{figure}[t]
    \centering
    \includegraphics[width=0.8\linewidth]{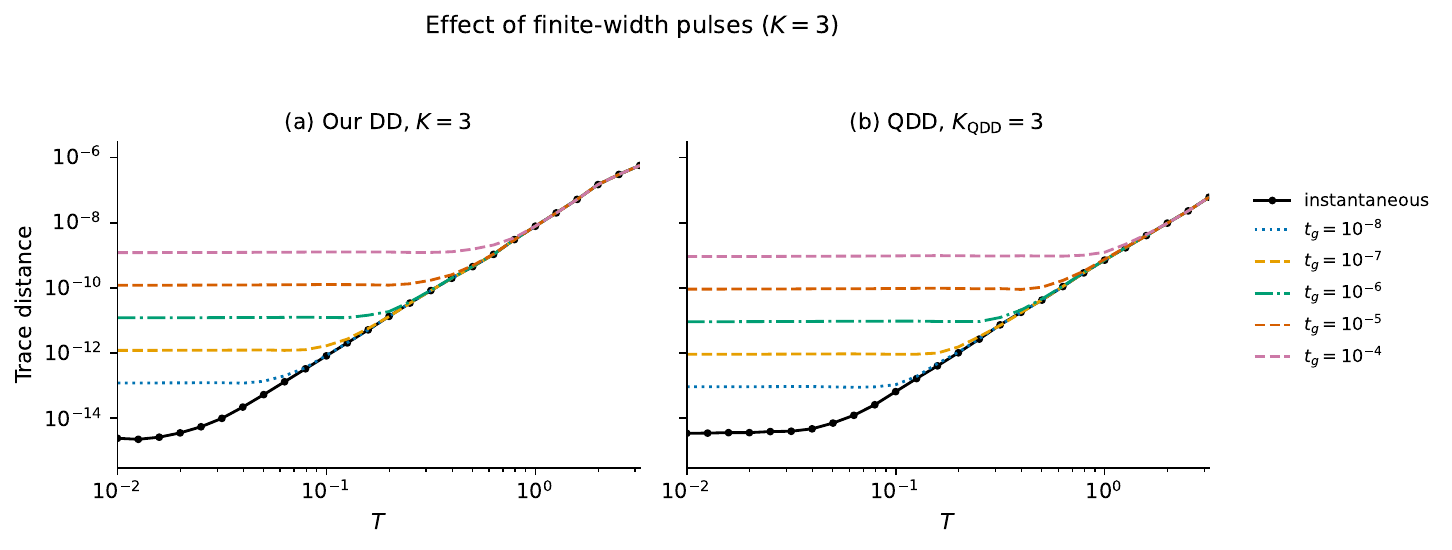}
    \caption{{Finite-pulse-width effects of (a) our optimized $K=3$ DD sequence and (b) QDD with $K_{\mathrm{QDD}}=3$. Each instantaneous pulse is replaced by a centered rectangular $\pi$ pulse of duration $t_g$. The same parameters as in Fig.~\ref{fig:trunc_DD} are used.}}
    \label{fig:finite-width}
\end{figure}

\subsection{Effect of finite pulse widths}

Our analysis assumes ideal instantaneous pulses. For pulses of duration $t_g>0$, the moment constraints do not generally cancel the finite pulse-width contributions, so finite pulse widths introduce an additional error floor. Writing the implemented $j$th pulse as $\widetilde P_j=E_jP_j$, where
\begin{align}
E_j=\mathcal{T}\exp\left(-i\int_0^{t_g}\widetilde H_j(t)dt\right),
\qquad
\|\widetilde H_j(t)\| = \mathcal O(J),
\end{align}
we have
\begin{align}
\|\widetilde P_j-P_j\|
= \|E_j-I\|
\le \int_0^{t_g}\|\widetilde H_j(t)\|dt
= \mathcal O(Jt_g).
\end{align}
Hence
\begin{align}
\|\widetilde U_{\mathrm{fw}}(T)-U_{\mathrm{ideal}}(T)\|
\le \sum_{j=1}^{L}\|\widetilde P_j-P_j\|
= \mathcal O(JLt_g).
\end{align}
Comparing with $ \mathcal O(JT^{K+1})$ yields
\begin{align}
T_c\sim (Lt_g)^{1/(K+1)}.
\end{align}
}

{We numerically test our DD sequence and QDD by replacing instantaneous pulses with rectangular pulses of finite duration $t_g$. As shown in Fig.~\ref{fig:finite-width}, both sequences develop an error floor that scales linearly with $t_g$, consistent with our analysis.}

{
\section{Suppressing $2$-local noise: high-order chromatic Hadamard DD}
%\lk{This newly added section addresses Referee A’s major concerns about applications to 2-local noise, particularly using CHaDD, as well as Referee B’s criticism that the example is too ``limited''. In fact, it also partially resolves one of the open problems raised in the CHaDD paper!} 

In this section, we move beyond 1-local system-bath couplings and consider decoupling 2-local couplings. For qubits placed on the vertices of a graph, with edges representing system-bath couplings, chromatic Hadamard dynamical decoupling (CHaDD) \cite{brown2025efficient} was recently proposed as a \textit{first-order} DD scheme whose sequence length scales with the graph’s chromatic number, making it more efficient than previous schemes. Our main result, Theorem~\ref{thm:existence}, applies directly to this setting and implies the existence of efficient higher-order CHaDD sequences in the weak-coupling regime. For 2- and 3-colorable graphs, which are exactly the examples considered in the original CHaDD paper \cite{brown2025efficient}, we explicitly construct such higher-order sequences and compare their performance. Hence, we partially answer one of the open problems posed in the original CHaDD paper, namely, the design of higher-order CHaDD schemes.

\subsection{Setup and CHaDD}
We consider a system of qubits placed on the vertices of a graph $\Gamma =(V,E)$, where $|V|=n$, and assume that the system-bath coupling is graph-local and at most two-body. Concretely, we write
\begin{align}
H = H_0 + H_{SB},
\end{align}
where $H_0$ denotes the part left invariant by the control, and
\begin{align}
H_{SB} &= \sum_{v\in V}\sum_{\alpha\in\{x,y,z\}}
\sigma_v^\alpha \otimes B_{v,\alpha} + \sum_{(u,v)\in E}\sum_{\alpha,\beta\in\{x,y,z\}}
\sigma_u^\alpha \sigma_v^\beta \otimes B_{uv,\alpha\beta}.
\label{eq:graph_local_noise}
\end{align}
Thus, arbitrary one-local terms are allowed on vertices, while two-local terms are allowed only on edges of $G$.

The \emph{chromatic number} $\chi$ of a graph is the smallest integer for which the vertex set $V$ can be partitioned into $\chi$ disjoint subsets $V = V_1 \sqcup \cdots \sqcup V_\chi$ such that no edge of $\Gamma$ has both endpoints in the same subset. Equivalently, a proper coloring of $\Gamma$ is a map
\begin{align}
c:V\to[\chi]:=\{1,2,\dots,\chi\}
\end{align}
such that $c(u)\neq c(v)$ whenever $(u,v)\in E$. We fix such a proper coloring and write
\begin{align}
V_a := c^{-1}(a), \qquad P_a := \prod_{v\in V_a} P_v, \qquad
a=1,\dots,\chi.
\end{align}
Since the color classes are disjoint, the operators $\{P_a\}_{a=1}^\chi$ mutually commute.

The CHaDD construction of Ref.~\cite{brown2025efficient} provides first-order decoupling schemes adapted to the chromatic structure of $\Gamma$, both for single-axis graph-local noise and for general graph-local noise. Applying Theorem~\ref{thm:existence} to the corresponding CHaDD decoupling groups yields the following higher-order version.

\begin{corollary}[High-order CHaDD]
\label{cor:high-order-chadd}
Assume that $H_0$ is invariant under the CHaDD protocols. Then, for every integer $K\ge 1$, there exists a piecewise-constant DD sequence such that
\begin{align}
M_{\alpha,m}=0,
\qquad
m=0,1,\dots,K-1,
\end{align}
for all Pauli terms appearing in $H_{SB}$. Consequently, in the weak-coupling regime,
\begin{align}
\|U(T)-U_0(T)\| = \mathcal O(JT^{K+1})+ \mathcal O(J^2T^2).
\end{align}
Moreover, the number $L$ of pulses can be chosen so that
\begin{align}
L\le
\begin{cases}
%\bigl(2^{\lceil \log_2(\chi+1)\rceil}-1\bigr)K \le 
(2\chi-1)K,
& \text{if } H_{SB} = \displaystyle \sum_{v\in V}\sigma_v^z\otimes B_{v,z} + \sum_{(u,v)\in E}\sigma_u^z\sigma_v^z\otimes B_{uv,zz}, \\
%\bigl(N_{\mathrm{ma}}(\chi)-1\bigr)K \le 
(6\chi+9)K,
& \text{for general } H_{SB} \text{ of the form \eqref{eq:graph_local_noise}}. \end{cases}
\end{align}
%where $N_{\mathrm{ma}}(\chi)$ denotes the multi-axis CHaDD sequence depth, satisfying $3\chi+1\le N_{\mathrm{ma}}(\chi)\le 2(3\chi+5)$. 
\end{corollary}

Hence, Corollary~\ref{cor:high-order-chadd} partially resolves an open problem posed in Ref.~\cite{brown2025efficient}, namely, the construction of higher-order CHaDD schemes using nonuniform pulse intervals.

\subsection{Examples}
We now consider the same two examples studied in the original CHaDD paper~\cite{brown2025efficient}, namely the single-axis setting for $\chi=2$ and $\chi=3$. In both cases, the relevant decoupling group is isomorphic to $\mathbb Z_2^2$, and the induced conjugation characters on the relevant error sectors coincide with the three nontrivial characters of the single-qubit Pauli decoupling group. Consequently, the sign triples $\mathbf s_\ell$ and hence the normalized moment-cancellation constraints are exactly the same as in the single-qubit general-decoherence problem!

\subsubsection{$\chi=2$: bipartite graphs}
Let $\Gamma$ be bipartite, with vertex partition $V=A\sqcup B$, and define
\begin{align}
g_x := X_A, \qquad g_z := X_B, \qquad g_y := g_xg_z = X_AX_B.
\end{align}
Then the decoupling group found in \cite{brown2025efficient} is
\begin{align}
G=\{I,g_x,g_y,g_z\}\cong \mathbb Z_2^2.
\end{align}
Under the identification
\begin{align}
I \leftrightarrow I, \qquad g_x \leftrightarrow X, \qquad g_y \leftrightarrow Y, \qquad g_z \leftrightarrow Z,
\end{align}
the relevant one- and two-body error sectors, namely $Z_v$ for $v\in A$, $Z_v$ for $v\in B$, and $Z_uZ_v$ for $(u,v)\in E$, realize exactly the three nontrivial characters of the single-qubit Pauli decoupling group. Therefore the control frames $f_\ell$, sign triples $\mathbf s_\ell$, and normalized moment constraints are identical to those in the single-qubit problem. Hence the same optimized interval vectors $\{\Delta\tau_\ell\}$ apply verbatim after the substitution
\begin{align}
X \mapsto X_A, \qquad Z \mapsto X_B.
\end{align}
Equivalently, the high-order sequence is the alternating pulse pattern
\begin{align}
X_A \to X_B \to X_A \to X_B \to \cdots,
\end{align}
with the same inter-pulse intervals as in Table~\ref{optim_SM}. This sequence is used to produce Fig.~\ref{fig:chadd_3x3} in the main text.

\subsubsection{$\chi=3$: three-colorable graphs}
Let $\Gamma$ be $3$-colorable, with color classes $V_1,V_2,V_3$, and define
\begin{align}
g_x := X_{V_1}X_{V_3}, \qquad
g_z := X_{V_1}X_{V_2}, \qquad
g_y := g_xg_z = X_{V_2}X_{V_3}.
\end{align}
Then the decoupling group found in \cite{brown2025efficient} is again
\begin{align}
G=\{I,g_x,g_y,g_z\}\cong \mathbb Z_2^2.
\end{align}
Under the same identification with $\{I,X,Y,Z\}$, the one-body sectors $Z_v$ for $v\in V_1,V_2,V_3$ realize the three nontrivial characters of the single-qubit Pauli decoupling group. Moreover, since every edge joins two distinct color classes, each two-body term $Z_uZ_v$ with $(u,v)\in E$ reproduces one of these same three characters according to the color pair of $(u,v)$. Therefore the induced sign patterns are exactly the same as in the single-qubit setting, and hence the control frames $f_\ell$, sign triples $\mathbf s_\ell$, and normalized moment constraints are unchanged. Thus the same optimized interval vectors $\{\Delta\tau_\ell\}$ may be reused after the substitution
\begin{align}
X \mapsto X_{V_1}X_{V_3}, \qquad
Z \mapsto X_{V_1}X_{V_2}.
\end{align}
Equivalently, the high-order sequence is the alternating pulse pattern
\begin{align}
(X_{V_1}X_{V_3}) \to (X_{V_1}X_{V_2}) \to (X_{V_1}X_{V_3}) \to (X_{V_1}X_{V_2}) \to \cdots,
\end{align}
with the same inter-pulse intervals as in Table~\ref{optim_SM}.
}

\end{document}